\documentclass[smallextended]{svjour3}

\usepackage{color}
\usepackage{amsmath}
\usepackage{amssymb}
\usepackage{subfigure}
\usepackage{graphicx}
\usepackage{multirow}

\graphicspath{{Figures/}}
\DeclareGraphicsExtensions{.png}

\newcommand{\E}{\mathbb{E}}

\newcommand{\rmk}[1]{{\color{blue} #1}}

\newtheorem{thm}{Theorem}[section]

\newtheorem{prop}[thm]{Proposition}
\newtheorem{lem}[thm]{Lemma}
\newtheorem{defi}[thm]{Definition}
\newtheorem{rem}[thm]{Remark}

\def\Esp{{\mathbb{E}}}

\def\calJ{{\cal J}}
\def\calQ{{\cal Q}}

\def\cqfd{$\square$}

\def\LB{\lambda_{\rm buy}}
\def\LS{\lambda_{\rm sell}}
\def\RB{R_{\rm buy}}
\def\RS{R_{\rm sell}}
\def\PB{p^{\rm buy}}
\def\PS{p^{\rm sell}}
\def\one{\mathbf{1}}

\def\CP{$C\!\rightarrow\!P$}
\def\PC{$P\!\rightarrow\!C$}

\titlerunning{Efficiency of the Price Formation Processs: a Mean Field Game Analysis}
\authorrunning{A Lachapelle, JM Lasry, CA Lehalle, PL Lions}
\title{Efficiency of the Price Formation Process in Presence of High Frequency Participants: a Mean Field Game Analysis } 
\author{  Aim\'e Lachapelle${}^*$,  Jean-Michel Lasry${^+}$, Charles-Albert Lehalle$\dagger$ and Pierre-Louis Lions$\ddagger$}
\institute{%
${}^*${MFG Labs},
${^+}${CEREMADE and Cr\'edit Agricole Corporate and Investment Bank},
$\dagger${Capital Fund Management (corresponding author, {\sf charles-albert.lehalle@cfm.fr}, tel: +33(0) 661 439 274)},
$\ddagger${CEREMADE and Coll\`ege de France}
} 
\date{}
\begin{document}
\maketitle


\abstract{This paper deals with a stochastic order-driven market model with waiting costs, for orderbooks with heterogenous traders. Offer and demand of liquidity drives price formation and traders anticipate future evolutions of the orderbook. The natural framework we use is mean field game theory, a class of stochastic differential games with a continuum of anonymous players. Several sources of heterogeneity are considered including the mean size of orders. Thus we are able to consider the coexistence of Institutional Investors and High Frequency Traders (HFT). We provide both analytical solutions and numerical experiments. Implications on classical quantities are explored: orderbook size, prices, and effective bid/ask spread. According to the model, in markets with Institutional Investors only we show the existence of inefficient liquidity imbalances in equilibrium, with two symmetrical situations corresponding to what we call {\it liquidity calls for liquidity}. During these situations the transaction price significantly moves away from the fair price. However this macro phenomenon disappears in markets with both Institutional Investors and HFT, although a more precise study shows that the benefits of the new situation go to HFT only, leaving Institutional Investors even with higher trading costs.}

\vspace{1cm}

\paragraph{Acknowledgments.}
{This work has been partially granted by the Cr\'edit Agricole Cheuvreux Research Initiative in partnership with the Louis Bachelier Institute, the Coll\`ege de France and the Europlace Institute of Finance. Authors thank Ioanid Ro\c{s}u for fruitful discussions about the orderbook model.} 

\section{Introduction}
With the recent changes in regulation on financial markets (MiFID, 2007, in Europe and Reg NMS, 2005, USA) the competition across trading venues favored the appearance of new trading rules, in a global attempt to capture most of the decreasing liquidity available in the post-2008 financial crisis world.
Trading venues thus proposed innovative ways to trade in electronic orderbooks (that have the favor of regulators and policy makers because of their native traceability):
\begin{itemize}
\item tiny tick sizes (i.e. the minimum price change between two consecutive quotes \cite{citeulike:5426005}) to attract automated market orders using SORs (Smart Order Routers \cite{citeulike:5177512}, \\ \cite{FOU06}),
\item low latency networks and matching engines, to allow high frequency players to decrease their exposure to market risk, in attempts to give them incentive to provide more liquidity \cite{citeulike:10491580},
\item maker/taker fee schedules to pay Liquidity Provider orders inserted in orderbooks, in order to attract liquidity,
\item creation of Dark Pools of various kinds (see \cite{citeulike:7500879}), to promote anonymous liquidity seeking so that large investors can continue to exchange blocks in an electronic manner,
\item size-priority and pro-rata matching rules \cite{citeulike:10881873} to complement the usual time-priority models,
\end{itemize}
are among these changes in market microstructure.

\bigskip
\noindent
The analysis of the efficiency of the emerging ecology of partially connected trading pools is questioned, especially since the flash crash \cite{citeulike:10491580}, \cite{citeulike:8676220} during which the US equity market has lost around 10\% of its value in 10 minutes, regaining it in 20 minutes. The resiliency of the liquidity provided by HFT (High Frequency Traders \cite{citeulike:8423311}) raised concerns. So did the dispersion of liquidity on such an heterogeneous network of pools.\\
Addressing these points is difficult because the market microstructure is not only a set of trading rules that could be studied statically, it changes with market participants behaviors, each of them trying to optimize her own utility function and anticipating others' moves (see \cite{citeulike:12047995} and \cite{citeulike:12099185} for more details).\\
This article provides an order-driven market modeling, where the volume of arriving flows is the risk source and where the key driver is the demand/offer of liquidity. \\
Yet only a limited number of papers have explored such models, the most notables being \cite{citeulike:6253089} (modeling the orderbook queue dynamics) and \cite{citeulike:12303370} (empirically studying orderbook data to extract the main components of the dynamics). The present paper can be seen as a very good complement to those two very interesting ones: with an accurate economic modeling on the one hand, and empirical results on the other hand. The MFG approach links them together since we provide for instance Partial Derivative Equation formulations (compatible with the Fokker-Planck equation described in the second paper) arising from a structural modeling (compatible with the modeling of the first one). The problem is very complex, and the topic deserves for more studies and publications. Indeed, such liquidity models involve a very large number of traders who arrive and leave the system at different times and strategically interact. Such components lead necessarily to complex situations. 

\bigskip
\noindent
In our model we consider smart traders (we call them {\it players} as soon as we use the game theory environment) that arbitrate between limit and market orders. That is they have to choose between the immediate transaction price and expected later transaction prices.  The dynamic model is in continuous time, in infinite horizon. Since patience is at the heart of our model, the present approach belongs to the family of waiting cost order-driven market models. Closely related papers are the work of \cite{fou05}, and the more recent paper by \cite{citeulike:6253089}. The former is the seminal waiting cost based model (as opposed to asymmetry information models) in discrete time. The latter is a continuous time approach where traders have the possibility to cancel their orders for free. This late assumption greatly simplifies the problem and allows the author to describe the equilibrium in an elegant manner. \\
Like in \cite{citeulike:6253089}, we use a continuous time model with Poisson processes used to model newcomers' arrivals. 
Nevertheless our model present some important dissimilarities. First the patience structure of traders is more endogenous since no cancellation of orders are permitted. Choices made by the players are thus irreversible and traders' anticipations of future events become a core issue. Their is a deep impact on the equilibrium equations: the problem becomes nonlinear. Secondly, the goal of our paper is to study the case with heterogenous traders, in particular to model the interactions of Institutional Investors and High Frequency Traders. In \cite{citeulike:6253089} several types of traders are considered, but the strategical arbitrage between market and impact orders is allowed only for one of the types. This is not the case in our model, where all types make choices.
\bigskip
\noindent
Game theory is necessary as soon as markets are incomplete. When markets are complete, strategy is unnecessary and the only task agents have to perform is to optimize in regards to the price. Order-driven markets are by essence incomplete since the source of risk is the random arrival of traders, and it is impossible to hedge this risk because choices made by traders are irreversible (note that the more realistic modeling where the modeler considers costly cancellation of orders must lead to a similar incompleteness).
Consequently we are convinced that game theory offers a proper framework.

Mean Field Games (MFG monotone systems, as detailed in the next section) are the suitable class of games that naturally allow to take into account the specific components of the order-driven market we consider, that is: a continuum of anonymous players, irreversibility of the actions, recursivity (anticipation of future prices). 
The resulting dynamics is thus a mix of backward-driven behaviors (based on actualized anticipations of future values of trades) and forward-driven ones (resulting from the immediate actions taken by agents). 
The MFG framework has been built to capture this two way dynamics, therefore this paper uses it to render the dynamics of a stylized orderbook, allowing to obtain results on different market configurations.
In the paper we introduce a new kind of mean field games in which players take one strategic decision at their arrival into the game.

\bigskip
\noindent
The paper is organized as follows. In section \ref{sec:MFG} we provide a quick introduction to Mean Field Game theory. In section \ref{sec:model} we introduce the modeling approach. We start with a one-sided orderbook as a base camp towards the two-sided orderbook exposed later. Section \ref{sec:pfp} and \ref{sec:eqanal} are dedicated to the introduction and  theoretical study of the recursive equations characterizing the equilibrium. Finally we conclude the paper in section \ref{sec:appli}, where we apply the model to several markets: markets with Institutional Investors only {\it versus} markets with both Institutional Investors and High Frequency Traders. 


\section{Mean Field Games: a quick introduction}\label{sec:MFG}
Mean Field Games (MFG for short) are a class of stochastic differential games with a continuum of agents. \\They have been introduced by \cite{citeulike:3614137}. Similar ideas have been introduced from an engineering viewpoint by \cite{HCM} and \cite{AJW}. From then on, MFG have known numerous developments and applications to various fields, mainly in economics \cite{lucasmoll,MFG_PP}, statistics \cite{PASG_NE,PASG_UL}, and human crowd behaviors \cite{LW}. The mathematics and numerics of MFG have been widely developed. Most of the mathematical tools for MFG have been the purpose of a 5 years course at Coll\`ege de France \cite{PLL_CDF}, and recent developments are described in \cite{CLLP} from an analysis viewpoint and \cite{CDL} and \cite{citeulike:12478248} with a probabilistic approach.

\bigskip
\noindent
In the continuum, agents are atomized, which means that their influence on the global state is reduced to nil. In economics, this aspect has to be linked with the notion of {\it price taker} agents as opposed to the case of a {\it price maker} monopolist for instance. This is precisely the whole continuum that makes the equilibrium. \\
The nil influence may have other sources than the presence of infinitely many players in the game. Indeed, a game with stochastic continuous entries and exits of players leads to the same property. This will be in particular the case of the model we propose in the present paper, where we consider Poisson entries and exits. 

\bigskip
\noindent
The information consists of a measure on the space of states $S$. Being a measure, it is often denoted by $m$ in the literature, but to be consistent with the notations of the model developped in the next sections, we rather call it $x$. Then $x(s)$ quantifies the density of agents having state  $s$.\\
In a MFG, players individually optimize (by choosing actions) their expected pay-off, considering the evolution of the global dynamic of the collectivity as an observable parameter (and they {\it anticipate} its evolution). Simultaneously, the statistical evolution of the collective dynamic follows from the individual optimal behaviors. The equilibrium takes place as soon as the anticipated evolution coincides with the statistical evolution.

\bigskip
\noindent
A core characteristic of Mean Field Games is that they are anonymous games. This notion is well known and means that the game is invariant for any permutation of the players. In other words, the players are not labelled. This assumption is very natural in complex systems involving numerous players.\\
Mean Field Games are approximations of anonymous games with finitely many players. But things are getting much simpler in MFG. The strategical powerlessness of individuals (i.e. the atomized characteristic of players) dramatically shrinks the traditional complexity (materialized by numerous coupling of the equilibrium equations) of $N$-player games, which is well-known as being Achilles' heel of classical stochastic differential games. Players interact with others only {\it via} the global state of the collectivity.

\bigskip
\noindent
In $N$-player stochastic differential games, each player $i$ optimizes her value function $u_i$, depending upon every individual states of agents (including herself). The equilibrium is then characterized by a complex system of coupled differential equations. \\
In a Mean Field Game, the $N$ value functions become a single value function $\mathcal{U}$ depending upon the the state $s$ of a generic player and the density $x$ of the continuum.\\
The MFG equilibrium is then characterized by a {\it master equation} verified by $\mathcal{U}$.
The master equation is in general very tricky and mathematically challenging. Their is a natural classification of cases in term of risk structure.
\begin{itemize}
\item Individual risk: in this case, the stochasticity of each player's dynamic is independent of each other. This particular case was firstly introduced. A major simplification is that the value function does not depend on the density $x$, but only on the state $s$. Consequently the master equation reduces to a system of two coupled partial differential equations having a forward-backward structure. The dynamic of the collectivity is deterministic.
\item Shared risk: here the only risk that agents face is common to all of them. When agent's space $S$ is finite (that is $x:= (x_1, ..., x_M)$), then the value function can be discretized 
\[\mathcal{U}:= (u_j), \; u_j(x_1,\dots,x_M), \; j=1,\dots,M.\]
This class of cases have been deeply investigated and is referred to as the case of {\it monotone systems} (see \cite{PLL_CDF}). The monotone system takes the following form:
\begin{equation}\label{eq:MFGsyst}
0=-r u_j - \sum _{k=1}^N \alpha_k(u,x) \frac{\partial u_j}{\partial x_k}+\beta_j(u,x), \; \mbox{ for } j=1,...,M, 
\end{equation}
where
\begin{equation*}
\begin{split}
& u \rightarrow \alpha_j(u,x) \mbox{ is monotone for all } j\\
& u \rightarrow \beta_j(u,x) \mbox{ is monotone for all } j.
\end{split}
\end{equation*}
We will see later the PDE of our orderbook model falls into this class of MFG.\\
Note that there is also a time dependent version of (\ref{eq:MFGsyst}) with a time derivative term added.
\item Mix models: some classes of cases that mix both shared and individual risks are needed for economic modeling (e.g. for solving the Krussel-Smith problem \cite{citeulike:6055723}). 
\end{itemize}

\section{Model}\label{sec:model}

The stylized orderbook used here is a two-sided one.
We start with a simple single-queue model as a base camp towards the two-sided one that is exposed later.

\subsection{A simple single-queue model with anticipations}\label{sec:onesided}
The purpose of introducing first a single queue model is didactic and does not aim at directly providing insights on orderbook modeling. However we believe this single queue is the occasion to introduce some key concepts, such as endogenous strategic entries of agents that anticipate the future. Consequently, sellers entering the system are also called players since we locate the modeling approach in the game theoretic framework (agents perform actions optimizing their respective pay-off).\\
In particular, when new sellers arrive, they look at the queue size and decide whether to enter the queue or not (action), after considering their expected pay-off (value function assessment).

\bigskip
\noindent
With this simplified model we introduce anticipatory behaviors  in a very stylized one-sided orderbook, where patient sellers arrive at exogenous Poisson rate and where the arrival rate of impatient buyers increases as soon as the queue size increases. We will finally use it to provide insights on the modeling of distinct execution protocols, namely process sharing and First In First Out protocol.

\paragraph{The model. } 
The arrival rate of players is continuous and stochastic. In this simplified model, it is exogenous. 
\begin{itemize}
\item As usual, they arrive following a Poisson process with intensity $\lambda$.
\item Impatient buyers arrive at rate $\mu(x)\geq 0$, a given increasing function of $x$; i.e. the more patient sellers in the queue, the higher arrival rate of impatient buyers.
\item The unit size of an order in the queue is $q$. The queuing discipline is a process sharing one (with no priority), i.e. individual service in a queue of size $x$ is worth $q/x$. In terms of trading rules, one may think about a pro-rata one \cite{Field2008rata}.
\item The pay-off gained by a player per unit of order is a nonnegative decreasing function of the queue size: $P(x)$. Typical cases are $P(x) := \, p > 0$ and $P(x) = 1/x$. On the other hand, there is a cost $c$ of waiting in the queue. 
\end{itemize}
 
 \smallskip
 \noindent
Now, as usual in game theory, there is a value function $u$ for any player. The value function depends upon the queue size $x$. It is the expected Profit \& Loss (P\&L) of a player entering the queue. Note that we assume that agents are risk neutral and that their reservation utility is set to 0, which means that an agent decides to enter the queue as soon as the value function is positive: $u(x) > 0$.
\\

\noindent
The value function dynamic comes from an infinitesimal expression of events impacting it:
\begin{itemize}
\item a newcomer enters the queue as soon as $u(x)>0$ (remember $u$ is the ``\emph{expected value received if you enter the queue}'').
\item in the scope of this toy model, the queue is consumed by an exogenous Poisson process of intensity $\mu(x)$.
Each time an order already waiting in the queue is partially executed (according to a prorata rule): its owner will sell $q/x$ shares for a price $P(x)$. The new expected value for a participant waiting in the queue in this case is thus $q/x\cdot P(x) + (1-q/x)\cdot u(x-q)$ (i.e. the first part of the expression comes from the sell of $q/x$ shares and the second one from the expected value of the queue that is now of size $x-q$).
\item in all other cases, the expected value does not change.
\item the waiting cost is proportional to $q$ (the size of the orders); it decreases the expected value of $u$ by $c\,q\,dt$, where $dt$ is the time unit.
\end{itemize}

\paragraph{A MFG formalization: 1. The Control.}
A subtle aspect of the MFG is players enter into the game following $N^\lambda$, and take the decision to stay in the queue or to leave the game (paying a reference price of zero). 
The natural notation would be 
\begin{itemize}
\item the index of an anonymous player $i$ is $i:=N^\lambda$,
\item its control $\delta^i$ it naturally deduced from its value function $u_i(x)$ (the value of staying in the queue). As soon as the value of staying in the queue is greater than paying the reference price, it solves the control problem of the agent:
  \begin{equation}
    \label{eq:ctrl:1D}
    {\cal U}^i(x_t):= \max_{\delta^i\in\{0,1\}} \delta^i u_i(x).
  \end{equation}
The solution is hence easy to express:
$$\delta^i=\one_{\{u_i(x)>0\}}.$$
\item note that $x$, the size of the queue, is our \emph{mean field}. Thanks to it $u(x)$ is shared by all the players.
\end{itemize}
The mean field $x_t$ evolves according to a stochastic differential equation:
\begin{equation}
  \label{eq:dx:1D}
  dx_t = \left( dN^{\lambda }_t \delta^i_t - dN^{\mu(x)}_t\right)q,
\end{equation}
with the notations $dN^{\mu(x)}_t$ for the queue-consuming point process.

We see the mean field dynamics involves the value function $u_i$ and no more the control once we inject the solution of the control problem in it (namely $\delta^i=\one_{\{u_i(x)>0\}}$):
$$dx_t = \left( dN^{\lambda }_t \one_{\{u_i(x)>0\}}- dN^{\mu(x)}_t\right)q.$$


\paragraph{A MFG formalization: 2. Definition of the cost function.}
The value function the $i$th agent wants to minimize is driven by the following running cost
\begin{equation}
  \label{eq:dJ:1D}
  dJ(x_t)=\left[ \frac q x_t P(x_t) +(1 - \frac q x_t ) J(x_t-q) \right] dN^{\mu(x)}_t - cq\, dt.
\end{equation}
  The additive waiting costs are compatible with the very short time scale having a sense for orderbook dynamics\footnote{  It can be noted here that another cost function $\calJ$ could be defined here as:
  $$d\calJ(x_t)=\left[\omega(q,x_t) P(x_t) +(1 -\omega(q,x_t) ) \calJ(x_t-q) \right] dN^{\mu(x)}_t - cq\, dt,$$
  where $\omega(q,x_t)$ is a random variable taking value 1 with a probability $q/x$ and 0 otherwise. 
  In such a case, instead of a prorata rule, we will have a trading rule for which an order is fully executed with a probability $q/x$, or not at all. This case covers the trading model of \cite{citeulike:6253089}, in which the orderbook matching rule is FIFO (First In, First Out), but any agent can modify and reinsert his order at any time. In such a case the probability for one specific agent to be first in the queue (and thus be fully filled), is $q/x$.   

Since $\Esp dJ = \Esp d\calJ$, the emerging dynamics are the same.}.

With such a formalism, the value function can be defined as 
$$u_i(X)=\Esp \int_{t=0}^{T} J^u_i(x_t)\, dt$$
given $x_0=X$, with $T$ ``large enough'' at the intraday time sale.
Remind that in our specific class of MFG, the identity $i$ of the agent and the time are bound: $i=N^\lambda$, meaning players come into the game according to the point process $N^\lambda$.

\paragraph{A MFG formalization: 3. Expression of the mean field.}
Thanks to the mean field $x_t$, the value function can be anonymized:
$$u(x):=u_i(x)=\Esp \int_{t=0}^{T} J^u(x_t)\, dt, \; \forall i.$$

And thus the dynamics of the mean field is agent-agnostic too:
$$dx_t = \left( dN^{\lambda }_t \one_{\{u(x)>0\}}- dN^{\mu(x)}_t\right)q.$$

\paragraph{A MFG formalization: 4. Stationary equilibrium as a fixed point of the value function.}
Thanks to the previous steps we now look for the stationary value of $u$.
Below we detail the equilibrium equation for each probability event, giving birth in few paragraph to an ordinary differential equation describing the value function.

\begin{align}\label{eq:uSQM}
u(x) = &\;\; (1-\lambda \one_{\{u(x)>0\}}dt - \mu(x) dt) \cdot u(x) && \leftarrow \mbox{ \textcolor{blue}{nothing happens}}  \\
&+ \lambda \one_{\{u(x)>0\}} dt \cdot u(x+q) && \leftarrow \mbox{ \textcolor{blue}{new queue entrance}} \nonumber \\
&+ \mu(x) dt \cdot \Big( \frac q x P(x) +(1 - \frac q x ) u(x-q)\Big) && \leftarrow \mbox{ \textcolor{blue}{service}} \nonumber\\
& - c q\, dt && \leftarrow \mbox{ \textcolor{blue}{waiting cost}} \nonumber
\end{align}
 
 \noindent
We can perform a Taylor expansion for small $q$ in the discrete equation above. In this way we derive the following differential equation:
\begin{equation*}
\begin{split}
0 = \frac {\mu(x)} x (P(x\!)-u)\! -\! c\! +(\lambda \one_{\{u>0\}} \!-\mu(x)) u'+\! \textcolor{blue}{q \Big(\frac 1 2 (\lambda \one_{\{u>0\}} \!-\! \mu(x)) u''\! +\! \frac{\mu(x)} x u'\Big),}
\end{split}
\end{equation*}
where the second order term is the last one (blue term).

\paragraph{First order analysis.}
Before approximating numerically the solution to (\ref{eq:uSQM}), we propose to get some insights on the shape of the solution by doing a first order analysis. More precisely, the solution to the queuing system described above is characterized by the sign of the value function $u$. Consequently we are interested in finding potential sign switching points of $u$.\\
The core modeling ingredient is the value of the Poisson arrival rate $\lambda$ relative to $\mu(x)$. \\
For the first order analysis we look at the first order equation:
\begin{equation}\label{eq:edoSQM}
\begin{split}
0 =  \frac {\mu(x)} x \Big(P(x)-u(x)\Big) - c +\Big(\lambda \one_{\{u(x)>0\}} -\mu(x)\Big) u'(x).
\end{split}
\end{equation}
\begin{rem}
Let us remark that equation (\ref{eq:edoSQM}) corresponds to a trivial shared risk Mean Field Game monotone system with $N=1$, as described in the previous section. Note that in the framework of this model, the {\it mean field} aspect does not come from the continuum of agents (for every instant, the number of players is finite), but rather to the stochastic continuous structure of entries and exits of players.
\end{rem}
\noindent
Now we look at the case where the stylized limit orderbook presented here has an {\it infinite resiliency}, meaning once the orderbook is partially consumed by a marketable order, the remaining liquidity rearranges itself to fill the gap.
Moreover, we will consider the non degenerated case where sellers arrive at rate $\lambda$, larger than the exogenous consuming rate $\mu(x),$ for all $x$. 

\paragraph{An example with anticipatory behavior.}
Assume the arrival rate of buyers has the specificity to take two values: 
\begin{itemize}
\item a low value $\mu_1$ below a certain queue size threshold $S$,
\item a higher value $\mu_2 \; (\; \mu_2>\mu_1)$, above the threshold $S$.
\end{itemize}
As a function depending upon the queue size variable $x$, it reads:
\[\mu (x) = \mu _1 \one _{x < S} + \mu _2  \one _{x \geq S}, \; 0 \leq \mu_1 < \mu_2.\]
Here there are at least two points where $u$ changes sign:
\begin{equation}\label{eq:2pts}
x_1^* = \mu_1 P(x_1^*)/c \mbox{ and }  x_2^* = \mu_2 P(x_2^*)/c. 
\end{equation}
Figure \ref{fig:SQM} shows the plot of the solution (numerical approximation of the solution to equation (\ref{eq:uSQM})) for a certain set of parameters (for $P$ constant). We can see that the first switching point is close to the first order approximation $x^*_1$, while the last sign switch significantly deviates from the first order approximation $x^*_2$. It means that higher order terms have a non-negligible effect.\\
But most importantly, we observe that there is another sign switch strictly below the threshold $S$. The existence of such a switching point means that players anticipate the improved service before the threshold is reached. Indeed, their value function becomes positive meaning that players enter the queue strictly before the improved service starts. This is why we talk about an {\it anticipation switching point}. Consequently, we can conclude that at the equilibrium, the strategical players adopt an anticipatory behavior.

\begin{figure}[!h]
  \begin{center}
    \includegraphics[width=1\linewidth]{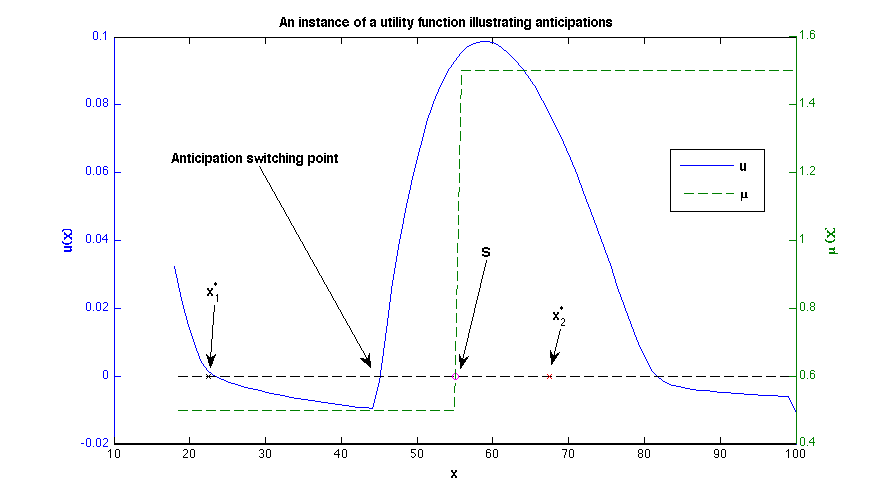}
  \end{center}
\caption{Here we notice that there is a point strictly before $S$ where $u$ switches from negative to positive. It means that players anticipate improved service $\mu_2$ and therefore are newly interested in entering the queue.}
\label{fig:SQM}
\end{figure}

\paragraph{First In First Out model.}
Finally we want to show that our approach allows to model distinct execution processes, and how the resulting equilibrium equations are impacted.\\
To do so, we consider the First In First Out (FIFO) protocol. This is the only change we make in  the model. To consider such a priority protocol, we have to introduce a new variable $z$ denoting the position of a trader in the queue of size $x$. Consequently the problem becomes bi-dimensional.\\
The equation becomes:
\begin{align}\label{eq:uFIFO1}
u(z,x) = &\;\; (1-\lambda \one_{u(x,x)>0}dt - \mu(x) dt) \cdot u(z,x) && \leftarrow \mbox{ \textcolor{blue}{nothing happens}}  \\
&+ \lambda \one_{\{u(x,x)>0\}} dt \cdot u(z,x+q) && \leftarrow \mbox{ \textcolor{blue}{new queue entrance}} \nonumber \\
&+ \mu(x) dt \cdot  u(z-q,x-q) && \leftarrow \mbox{ \textcolor{blue}{execution of the first order}} \nonumber\\
& - c q dt && \leftarrow \mbox{ \textcolor{blue}{waiting cost}}, \nonumber
\end{align}
in the domain $q<z<x$, and the boundary condition for $z=q$ is:
\begin{align}\label{eq:uFIFO2}
u(q,x) = &\;\; (1-\lambda \one_{u(x,x)>0}dt - \mu(x) dt) \cdot u(q,x) && \leftarrow \mbox{ \textcolor{blue}{nothing happens}}  \\
&+ \lambda \one_{\{u(x,x)>0\}} dt \cdot u(q,x+q) && \leftarrow \mbox{ \textcolor{blue}{new queue entrance}} \nonumber \\
&+ \mu(x) dt \cdot  P(x) && \leftarrow \mbox{ \textcolor{blue}{execution of the first order}} \nonumber\\
& - c q dt && \leftarrow \mbox{ \textcolor{blue}{waiting cost}}. \nonumber
\end{align}
System (\ref{eq:uFIFO1}-\ref{eq:uFIFO2}) can be easily solved numerically.

\subsection{The orderbook model}\label{sec:obmodel}

\paragraph{The matching mechanisms of order books.}
One of the roles of financial markets is to \emph{form prices} according to the balance between offer and demand. 
In modern markets, this mechanism takes place inside \emph{electronic order books where multilateral trading takes place}. They implement the following dynamic:
\begin{enumerate}
\item buyers and sellers can send electronic messages to a ``\emph{matching engine}''. These messages, called \emph{orders}, contain a side (``buy'' or ``sell''), a limit price and a quantity.
\item The matching engine contains a list of all pending orders it received in its memory. When it receives a new buy (respectively sell) order, it looks if pending sell (resp. buy) orders at a lower (resp. higher) price are available.
  \begin{itemize}
  \item If it is the case, it generates transactions between the owner of the incoming order and the owners of the compatible opposite orders, and removes the corresponding quantities in its list of pending orders;
  \item if the incoming order has a remaining quantity, it is inserted in the list of pending orders.
  \end{itemize}
\end{enumerate}
The list of pending orders waiting in the matching engine is called its ``\emph{limit order book}'' (LOB). 


 During the matching process, it is possible that the quantity of an incoming order does not match exactly the quantity made available at a compatible price (i.e. lower prices for a buy order and higher prices for a sell order) by opposite orders in the order book.
To handle such cases, matching engines need to implement a priority mechanism. The most used (see \cite{citeulike:10881873} for more details) are:
\begin{itemize}
\item \emph{time priority}: the ``oldest'' pending orders in the order book are matched first;
\item \emph{size priority}: the largest pending orders are matched first in case of competition between resting orders at the same price;
\item \emph{pro rata}: pending orders are matched for a fraction of their quantity proportionally to their relative size to the one of the whole queue (see \cite{Field2008rata}).
\end{itemize}
\smallskip

Each trading platform discloses its matching mechanism in detail to market participants in a \emph{rulebook} (like \cite{enx0406}).

\paragraph{Matching dynamics and trading styles.}
Market participants thus have to cope with rules of the matching engine they trade into while fulfilling their day-to-day goals.
Recent regulatory discussions raised questions on the potential negative interactions between the following classes of market participants in the same order book:
\begin{itemize}
\item \emph{Institutional investors}, that buy and sell large quantities of shares to manage their portfolios on the long term. They take the decision to buy or sell independently from the immediate state of the order book. They are in essence \emph{impatient} since they interact with other participants in the order book with the final goal to really buy or sell given quantities before a given deadline. They will not change their mind during the trading process given the state of the liquidity in the order book.
\item \emph{High Frequency Traders} are far more opportunistic. 
  Even if they do not have all the same behavior (see \cite{citeulike:11858957} for more details), they have in common the fact that: (1) they send very small orders to trading platforms, (2) they do it very often (i.e. at \emph{high frequency}), (3) they have no other reason to trade than the immediate state of the order book.
\end{itemize}
\smallskip

Concerns raised focused on the integrity of the price dynamics when so different participants are mixed in order books. The ``Flash Crash''  \cite{citeulike:8676220} has shown that liquidity glitches could cause large variations of prices formed in electronic order books with no fundamental reasons.
Academics studying the price formation process in order books usually name ``\emph{temporary market impact}'' the way prices temporally deviate from their stable value due to high consumption of liquidity (i.e. of pending orders) in an order book (see \cite{citeulike:4325901}, \cite{citeulike:5177397}).

Recent regulatory changes unexpectedly favored HFTs activity \cite{citeulike:12047995} (they are said to now be part of 70\% of transactions in the US, 40\% in Europe and 30\% in some Asian markets, like Japan).

\paragraph{Dedicating a model to study liquidity games in order books.}
The way market participants interact in order books is sophisticated, due to the fact that they continuously try to anticipate actions of other participants to take an adequate decision. Their classical dilemna is the following. On the one hand they want to trade as slow as possible to avoid to be detected nor consume liquidity too fast thus moving the price an unfavorable way (i.e. adverse selection costs). On the other hand they cannot afford to trade too slow to avoid to be exposed to adverse market moves (i.e. opportunity costs).

A large literature proposes mathematical frameworks for market participant to optimize their trading kinematics: first mean-variances approaches \cite{OPTEXECAC00}, then stochastic control ones \cite{citeulike:5797837}, and more recently stochastic algorithms have been designed to capture optimally liquidity at the smallest time scale \cite{citeulike:5177512}.
In all these approaches, each market participant tries to optimize her behavior assuming that the aggregation of other players is ``martingale'' in the sense that it is submitted to price moves and to some order books characteristics (like the volatility, the market depth, the intensity of orders reaching the matching engine, etc.) emerging from the activity of other participants without influencing it (in most cases a \emph{market impact function} is introduced, exogenously from the activity of other participants).

The MFG approach presented here takes into account the way strategies of market participants change the dynamics of the order book. It opens the door to more endogenous models.
The previous section is a simple illustration of this approach: the mean field is the state of the one-sided order book, and since each player implements an optimal strategy (in the sense that she values the time to wait in the queue and compares it to an immediate price to pay), it is possible to understand the dynamics of the value function $u$ shared by all market participants.

In this section we will go one step further: the consuming rate $\mu(x)$ of the one-sided order book (say it is the queue of sellers) of section \ref{sec:onesided} will be linked to the size of the queue $x$, but in an endogenous way: that is via optimal strategies followed by participants in the other queue (the one of buyers).
The flow consuming the selling queue is the one of buyers deciding on their side to pay immediately instead of waiting in the queue (of buyers). It will enable the emergence of coupled dynamics taking into account the states of the two queues.

To render a market impact effect, we will model the way impatient buy or sell orders consume the queue of sellers or buyers.
For the ease of presentation, in this paper we will consider that at our time scale the ``\emph{fair price}'' (that can be understood as a \emph{latent price} like in \cite{citeulike:8317402} or \cite{citeulike:9217792}, or as a \emph{fundamental price} like in \cite{ho1983dynamics}) does not change significantly. But the reader can note that extending this model pegging a diffusive behavior on this fair price will do no more than adding an Ito term to the considered dynamics.
Mean Field-inspired models at a largest time scale, targeting the understanding of the latent price dynamics have been already proposed, but not at the level of the order books (for instance in  \cite{citeulike:7621540}, the dynamics of a latent order book is submitted to an MFG like mechanism, but the realizations of the order book is modeled via a forward only scheme).

Hence, we introduce here a market-impact like relation between the size of the order book queues and the transaction prices around the \emph{fair price} $P$:
consuming a quantity $q$ of the queue of pending selling orders of size $Q^a$ will temporally move the price from $P$ to $P+\delta\cdot q/(Q^a-q)$.
Qualitatively, it implies an almost linear market impact with elasticity $\delta$ (i.e. $\delta$ can be compared to Kyle's lambda \cite{citeulike:3320208}).
Moreover, the modelled orderbook will have an \emph{infinite resiliency}: once liquidity is consumed in a queue, the remaining quantity will reshape itself to fill the created gap.

The details of the MFG model are exposed in the following sections. In short, it contains these following ingredients:
\begin{itemize}
\item market participants are able to act strategically, anticipating others' moves;
\item the dynamics of the two queues (patient buyers and patient sellers) are coupled thanks to the fact that the flow consuming each of them is provided by agents of the other side choosing to be impatient (either because they do not use a smart routing strategy, or because the outcome of their smart strategy is to send a market order);
\item market impact is introduced dynamically (related to the size of the queues), modifying the premium to be paid by impatient traders, thus influencing their choices.
\end{itemize}
%
Moreover, our order book model needs a priority rule, for simplicity reasons we will use a pro-rata rule (since it keeps the dimensionality of the model tractable).
As it will be seen later, it allows to render enough complexity to obtain meaningful results.

\paragraph{Rendering different trading styles in an order book model.}
To understand the features of our MFG model, we will first study its dynamics in a market with homogenous participants. Since we are in a MFG framework, it will render a continuum of agents, at this stage they share the same macroscopic parameters:
\begin{itemize}
\item the same messaging intensity $\lambda$,
\item the same size of orders they send $q$,
\item the same waiting cost $c$.
\end{itemize}
Beside, we enrich the model with one more feature: the use of SOR (Smart Order Router).
A Smart Order Router (see \cite{FOU06} for an efficiency study or \cite{citeulike:12047995} for a generic presentation) is a device containing a software dedicated to ``smartly route'' orders. In our model, only SOR users will be able to act strategically instead of being blindly impatient.

It can be considered that agents not using a SOR have an infinite waiting cost. Since institutional investors take decisions independently of the current state of the orderbook, it is realistic to consider that a fraction of them will not take time to implement sophisticated microscopic strategies on some of their orders.

The proportion of market participants using a SOR (i.e. not infinitely impatient market participants) will be parametrized thanks to a specific flow of intensity $\lambda^-$.
\\

\begin{table}[h!]
  \centering
  \begin{tabular}{l|c|c|}
                     & Instit. Investors & HFT \\\hline
    Order size & large & small \\\hline
    Speed & normal & fast \\\hline
    SOR & often used & always used \\\hline
  \end{tabular}
  \caption{Qualitative modeling of Institutional Investors and HFT.}
  \label{tab:paramIIHFT}
\end{table}
In a second stage we will mix heterogenous agents, with different behaviours summarized in Table \ref{tab:paramIIHFT}:
\begin{enumerate}
\item \emph{Institutional investors}, trading large quantities not using systematically a SOR;
\item \emph{HFT (High Frequency Traders)}, faster than the former participants, using smaller orders, more patient (in the sense that they bare a lower cost per share waiting in a queue), and all of them using a SOR.
\end{enumerate}


\paragraph{Transaction price.}
The \emph{market price} will be centered on a constant 
$P$. The \emph{market depth} is $\delta$, meaning that no transaction will take place at a price lower than $P-\delta$ or higher than $P+\delta$. The (time varying) size of the bid queue (waiting buy orders) is $Q^b_t$ and the size of the ask one (waiting sell orders) is $Q^a_t$.

When a market (buying) order hits the ask queue, the transaction price is $\PB$ and when the bid queue is lifted by a market (selling) order, the transaction price is $\PS$. The price takes into account instantaneous queue size adjustments depending upon the order size $q$.
\begin{equation}
  \label{eq:execprice}
  \PB_q(Q^a_t) := P + \frac{\delta q}{Q^a_t-q},\;\;\; \PS_q (Q^b_t):= P - \frac{\delta q}{Q^b_t-q}
\end{equation}
\noindent
Qualitatively, it means that the market impact is linear. Boundary conditions, to be introduced later, impose $Q^a_t, Q^b_t > q$, so that there is no definition problem of the transaction prices


\paragraph{Value functions.}
The value function for a trader submitting a buy order in the bid queue is $v(Q^a_t,Q^b_t)$ and the one of a sell order in the ask queue is $u(Q^a_t,Q^b_t)$. In the model agents have risk-neutral preferences, thus the utility functions coincide with price expectations.

\paragraph{Orders arrival rates.} We distinguish between SOR and non-SOR orders. The proportion of these two types of orders is exogenous, and set as an input of the model.

\bigskip
\noindent
Buy and sell SOR orders arrive according to two Poisson processes with intensity $\LB$ and $\LS$.
Several cases can be considered:
\begin{enumerate}
\item Homogeneous Poisson processes: 
\begin{equation}\label{eq:Prates}\LB  = \LS =  \lambda.\end{equation}
\item Heterogeneous (in space) Poisson processes  
\[  \LB = \lambda f(Q^b_t),\; \LS = \lambda f(Q^a_t),\]
where $f(x)$ is a decreasing function. Typical instances are $ f(x) = 1/x$, $f(x) = \one _{x\leq \bar Q}$ likewise.\\
\end{enumerate}
However, we will focus in this paper on the homogenous case.\\
Let us remark that the previous rates could be endogenized and set as the result of an optimization problem involving the utility functions, consequently depending upon the queue sizes $Q^{\bullet}_t$.
\bigskip

\begin{figure}[h!]
  \centering
  \includegraphics[width=.5\linewidth]{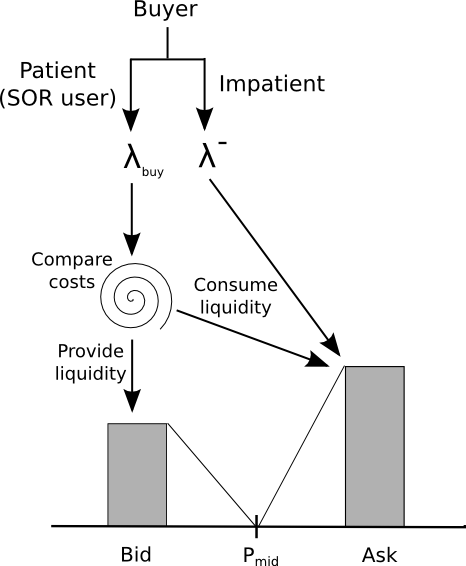}
  \caption{Idealized diagram of the decision tree of agents in the model.}
  \label{fig:digdecision}
\end{figure}

\noindent
Non-SOR orders (i.e. belonging to very impatient investors or traders) are always liquidity remover, with arriving rate $2\lambda^-$ (equally distributed between buyers and sellers).

\paragraph{Market participants decision processes.}
When a buy (resp. sell) order arrives, its owner has to make a \emph{routing decision} (see Figure \ref{fig:digdecision} for an idealized diagram of this process):
\begin{itemize}
\item if $v(Q^a_t,Q^b_t+q)<\PB(Q^a_t)$ (resp. $u(Q^a_t+q,Q^b_t)>\PS(Q^b_t)$) it is more valuable to route the order to the bid (resp. ask) queue (i.e. sending a limit order). In such a case the order will be a \emph{Liquidity Provider (LP)}.  We define symmetrically \emph{Liquidity Consumer (LC)} orders. This decision is formalized in the model by setting the variable $\RB^\oplus(v,Q^a_t,Q^b_t+q)$ to 1 when $v(Q^a_t,Q^b_t+q)<\PB(Q^a_t)$, and to zero otherwise:
\end{itemize}
  \begin{equation}
    \label{eq:routing}
    \begin{split}
    &\RB^\oplus(v,Q^a_t,Q^b_t+q):=\one_{v(Q^a_t,Q^b_t+q)<\PB(Q^a_t)}, \mbox{ \rmk{LP buy order}}\\
    &\RS^\oplus(u,Q^a_t+q,Q^b_t)\;:=\one_{u(Q^a_t+q,Q^b_t)>\PS(Q^b_t)}, \mbox{ \rmk{LP sell order}}.
   \end{split}
  \end{equation}
  \begin{itemize}
\item otherwise the order goes Liquidity Consumerly to the ask (resp. bid) queue to obtain a trade. It will be a liquidity remover in this case:
 \begin{eqnarray*}
  &&\RB^\ominus(Q^a_t\!,Q^b_t\!):=1-\RB^\oplus(Q^a_t\!,Q^b_t\!), \mbox{ \rmk{LC buy order}}\\
  &&\RS^\ominus(Q^a_t\!,Q^b_t\!)\;:=1-\RS^\oplus(Q^a_t\!,Q^b_t\!), \mbox{ \rmk{LC sell order}}.
  \end{eqnarray*}
  The price of such a transaction is $\PB$ (resp. $\PS$) as defined by equality (\ref{eq:execprice}). Note that we omit the dependence on $u, v$ when  it is unnecessary for the understanding of the equations.
\end{itemize}
We impose the following boundary conditions:
\begin{equation}\label{eq:bc}
\begin{split}
 \mbox{Min liquidity condition : }\RB^\oplus(r,Q^b_t) = 1, \; \RS^\oplus(Q^a_t,r) = 1, \; \forall r \leq q, \\
\mbox{Technical condition : }\RB^\oplus(Q^a_t,r) = 1, \; \RS^\oplus(r,Q^b_t) = 1, \; \forall r < q.
\end{split}
\end{equation}
In particular, conditions (\ref{eq:bc}) ensure that $(Q^a_0,Q^b_0)\geq (q,q) \Rightarrow (Q^a_t,Q^b_t) \geq (q,q), \; \forall t>0$.

\paragraph{A MFG formalization: 1. The Control.}
Like in Section \ref{sec:onesided}, we will adopt a more standard MFG formalism.
First of all note the identity of an agent $i$ has a one-to-one correspondance with the sum of the two Poisson processes $N^\LS$ and $N^\LB$ containing all the arrivals arrivals. When at time $t$ a $i$th agent enters the game, it can be a selling agent (in this case $dN^{\LS}_t=1$) or a buying agent  (in this case $dN^{\LB}_t=1$); in any case $i:=N^{\LS}_t+N^{\LB}_t$.

A selling agent $i$ can control its cost setting $\RS^\oplus=1$ (in such a case he will stay in the queue) or $\RS^\oplus=0$ (in such a case he will consume the other queue). For a buying agent, the control is $\RB^\oplus$.

The decision is taken to minimize his cost function:
\begin{itemize}
\item on the one hand the selling agent knows the immediate price if he consumes liquidity on the bid queue (it is $\PS_q (Q^b_t)$, defined by equality \eqref{eq:execprice});
\item on the other hand, by construction the expected value to wait in the queue is $u_i(Q^a_t + q,Q^b_t)$.
\end{itemize}

Like in the one queue toy model (equation \eqref{eq:ctrl:1D}), his optimal control is hence chosen to maximize the selling price:
$${\RS^\oplus}_{,i}:=\arg\max_\delta \delta \cdot u_i(Q^a_t + q,Q^b_t) + (1-\delta)\cdot \PS_q (Q^b_t).$$

\paragraph{A MFG formalization: 2. Definition of the cost function.}
%
The dynamics associated with this matching mechanism can be written:
\begin{itemize}
\item for the size of the ask queue $Q^a_t$ (it is equivalent to equation \eqref{eq:dx:1D} of the one queue toy model):
  \begin{equation}
    \label{eq:dQ}
      dQ^a_t= \left( dN^{\LS} {\RS}^{\oplus}_{,i} - (dN^{\LB } {\RB}^\ominus_{,i'} + dN^{\lambda^-})\right) q,
  \end{equation}
  where $i$ is the identity of the selling agent taking a decision at $t$ (i.e. $i:=N^{\LS}_t$)
  and $i'$ is the identity of the buying agent taking a decision at $t$ (i.e. $i':=N^{\LB}_t$).
\item and for the running cost function at the ask (similarly to equation \eqref{eq:dJ:1D} of the one queue toy model):
  \begin{eqnarray}
    \label{eq:dJi}
      dJ^u_i(Q^a,Q^b)&=&\left[ \frac{q}{Q^a} \PB(Q^a) + \left( 1-\frac{q}{Q^a}  \right) J^u_i(Q^a-q, Q^b)\right]\\
      &&\;\; \cdot ( dN^{\LB } {\RB}^\ominus_{,i'} + dN^{\lambda^-} )- c_a q \, dt. \nonumber
  \end{eqnarray}
  The index $i'$ underlines the agents interacting with the cost function associated to wait on the ask side are the buying ones, and the index $i$ underlines the agents taking decision using this cost function to choose their control are sellers.
\end{itemize}
Again, with $T$ large enough, $u_i(\calQ^a, \calQ^b) = \Esp \int_{t=0}^T J^u_i(Q^a_t, Q^b_t)\, dt$ given $Q^a_0=\calQ^q, Q^b_0=\calQ^b$.

\paragraph{A MFG formalization: 3. Expression of the mean field.}
In this case the \emph{mean field} is two dimensional. It it made of the sizes of the two queues $(Q^a_t, Q^b_t)$. One can note the identity of the agents $i$ and $i'$ has no importance in equations (\ref{eq:dQ}) and (\ref{eq:dJ}), all the dynamics are summarized by $(Q^a_t, Q^b_t)$.

Thanks to this remark we can write the \emph{forward} dynamics of the mean field
\begin{equation}
  \label{eq:MFQ}\left\{
  \begin{array}{lcl}
    dQ^a_t/q&=& dN^{\LS} {\RS}^{\oplus} - (dN^{\LB } {\RB}^\ominus + dN^{\lambda^-}) \\
    dQ^b_t/q&=& dN^{\LB} {\RB}^{\oplus} - (dN^{\LS } {\RS}^\ominus + dN^{\lambda^-}) 
  \end{array}\right.
\end{equation}
in which we can plug the solution of the optimal control choices
\begin{equation}
  \label{eq:optC}\left\{
      \begin{array}{lcl}
        {\RS^\oplus}&=&\arg\max_\delta \delta \cdot u(Q^a_t + q,Q^b_t) + (1-\delta)\cdot \PS_q (Q^b_t)\\
        {\RB^\oplus}&=&\arg\max_\delta \delta \cdot v(Q^a_t + q,Q^b_t) + (1-\delta)\cdot \PB_q (Q^a_t)
      \end{array}\right. . 
\end{equation}

Again, thanks to the mean field the indices $i$ and $i'$ are no more needed. To be able to make the optimal choice, the agents have to solve the dynamics of the value function
\begin{equation}
  \label{eq:Jmean}\left\{
  \begin{array}{lcl}
    u(\calQ^a, \calQ^b) &=& \Esp \int_{t=0}^T J^u(Q^a_t, Q^b_t)\, dt\\
    v(\calQ^a, \calQ^b) &=& \Esp \int_{t=0}^T J^v(Q^a_t, Q^b_t)\, dt
  \end{array}\right. ,
\end{equation}
where $J^u$ and $J^v$ are now defined without any reference to the identity of the agent $i$ or $i'$;
definition (\ref{eq:dJi}) now becomes:
\begin{eqnarray}
    \label{eq:dJ}
      dJ^{u}(Q^a,Q^b)&=&\left[ \frac{q}{Q^a} \PB(Q^a) + \left( 1-\frac{q}{Q^a}  \right) J^u(Q^a-q, Q^b)\right]\\
      &&\;\; \cdot ( dN^{\LB } {\RB}^\ominus + dN^{\lambda^-} )- c_a q \, dt. \nonumber
  \end{eqnarray}
and $dJ^v$ is naturally defined a similar way.

The last step of the mean field game formalisation for our MFG orderbook is developed in the next section.

\paragraph{Remark about the matching process.} 
Before this last step, note that
the matching process is close to a pro-rata one \cite{Field2008rata}: in case of a liquidity consuming buy order of size $Q$ to be matched, all market participants having a quantity $q$ resting in the ask queue will obtain a transaction for a fraction $Q\cdot q/Q^a_t$ of its order at price $\PB(Q^a_t)$, the remaining quantity staying in the orderbook.\\
At a first glance one may think that this matching process will induce intricate terms in the equations, but in fact it will not since we only consider utilities by units of good transactions.

  \begin{itemize}
  \item The orderbook shape is assumed to be linear (in the price),
    meaning that if a newcomer decide to provide liquidity to the
    market, her order will be split proportionally to the liquidity
    already present in the book: the orderbook will remain linear in
    price with a higher slope.
  \item Hence when a Liquidity Consumer order occurs, it will
    partially fill all Liquidity Provider orders according to a proportional rule.
  \end{itemize}

\section{The PFP (Price Formation Process) dynamics}\label{sec:pfp}

\subsection{Stationary equilibrium as a fixed point of the value function: introducing the equations}

The fourth step of the MFG formalisation of our mean field game orderbook allows us to characterize an equilibrium via recursive equations of the expected value of future payoffs (value functions).
\def\sone{\rule{1em}{0pt}}

\begin{align}\label{eq:ua}
  \hspace*{.5cm} &    u(Q^a_t,Q^b_t) = \\ 
  &(1 - \LB dt-\LS dt - 2\lambda^- dt) \; u(Q^a_t,Q^b_t) && \leftarrow \mbox{\textcolor{blue}{ nothing}} \nonumber \\ 
  &+(\LS \RS^\ominus(u,Q^a_t+q,Q^b_t) + \lambda^-)dt \;  u(Q^a_t,Q^b_t-q)  &&\leftarrow \mbox{\textcolor{blue}{sell order, LC}} \nonumber \\ 
  &+\;\LS \RS^\oplus(u,Q^a_t+q,Q^b_t) dt \; u(Q^a_t+q,Q^b_t)  &&\leftarrow \mbox{\textcolor{blue}{sell order, LP}}\nonumber \\ 
  &+(\LB \RB^\ominus(v,Q^a_t,Q^b_t+q) + \lambda^-)dt \cdot \big[  &&\leftarrow \mbox{\textcolor{blue}{ buy order, LC}}\nonumber \\ 
  &\hspace*{0.07\textwidth} \underbrace{\frac{q}{Q^a_t} \PB(Q^a_t)}_{\mbox{\textcolor{blue}{ trade part (ask)}} }    + \underbrace{(1- \frac{q}{Q^a_t})\,  u(Q^a_t-q,Q^b_t)}_{\mbox{\hspace{0.03cm}\textcolor{blue}{ removing (ask)}}} \big] \nonumber \\
  & + \; \LB\, \RB^\oplus(v,Q^a_t,Q^b_t+q) dt \; u(Q^a_t,Q^b_t+q)  &&\leftarrow \mbox{\textcolor{blue}{buy order, LP}}\nonumber \\ 
  & - \; c_a q\,dt. &&\leftarrow \mbox{\textcolor{blue}{cost to maintain inventory}} \nonumber
  \end{align} 
\noindent
Symmetrically, we have :
\def\sone{\rule{1em}{0pt}}
\begin{align}\label{eq:ub}
  \hspace*{.5cm} &    v(Q^a_t,Q^b_t) = \\ 
  &(1 - \LB dt-\LS dt - 2\lambda^- dt) \; v(Q^a_t,Q^b_t) && \leftarrow \mbox{\textcolor{blue}{ nothing}} \nonumber \\ 
  &+(\LB \RB^\ominus(v,Q^a_t,Q^b_t+q) + \lambda^-)dt \; v(Q^a_t-q,Q^b_t)  &&\leftarrow \mbox{\textcolor{blue}{buy order, LC}} \nonumber \\ 
  &+\;\LB \RB^\oplus(v,Q^a_t,Q^b_t+q) dt \; v(Q^a_t,Q^b_t+q)  &&\leftarrow \mbox{\textcolor{blue}{buy order, LP}}\nonumber \\ 
  &+(\LS \RS^\ominus(u,Q^a_t+q,Q^b_t) + \lambda^-)dt \cdot [ &&\leftarrow \mbox{\textcolor{blue}{ sell order, LC}}\nonumber \\ 
  &\hspace*{0.07\textwidth} \underbrace{ \frac{q}{Q^b_t} \PS(Q^b_t)}_{\mbox{\textcolor{blue}{ trade part (bid)}} }    + \underbrace{(1- \frac{q}{Q^b_t})\, v(Q^a_t,Q^b_t-q)}_{\mbox{\hspace{0.03cm}\textcolor{blue}{ removing (bid)}}} \big] \nonumber \\
  & + \; \LS\, \RS^\oplus(u,Q^a_t+q,Q^b_t) dt \; v(Q^a_t+q,Q^b_t)  &&\leftarrow \mbox{\textcolor{blue}{sell order, LP}}\nonumber \\ 
  & - \; c_b q\,dt. &&\leftarrow \mbox{\textcolor{blue}{cost to maintain inventory}} \nonumber
  \end{align}

\noindent
Remind that $\RB$ and $\RS$ are functionals of $Q^a$ and $Q^b$ and also implicitly depends on $u$ and $v$. Of course the previous principles hold for $Q^a_t,Q^b_t> q$, which is always the case thanks to conditions (\ref{eq:bc}).
In the equations above, $c_a$ and $c_b$ are positive constants modeling the cost to maintain inventory per unit, that is the cost of never being processed once waiting in the queue.


\subsection{Symmetric case}
In the case where $\LS = \LB = \lambda$, and $c_a=c_b=c$, we have the following results. \\
For the sake of simplicity we will often use new notations for the queue size variables: $x$ and $y$ stand for $Q_a$ and $Q_b$.
\begin{lem}
\[\forall(x,y),  \; \RS^\oplus(u,x,y)= \RB^\oplus(2P-v,y,x) \]
\end{lem}
\noindent
This simple symmetry result is useful to get a necessary condition for the solution.
\begin{prop}
If system (\ref{eq:ua})-(\ref{eq:ub}) has a unique solution $(u,v)$, then
\[\forall(x,y), \;\;u(x,y) +P = P-v(y,x).\]
That is, $u$ and $v$ are antisymmetric up to the constant $P$.
\end{prop}
\begin{proof}
Take Equation (\ref{eq:ua}) then perform the change of variable $w(y,x) = 2P-u(x,y)$, then apply the previous Lemma, switch the roles of $x$ and $y$ and multiply by $-1$. Then you get equation (\ref{eq:ub}), hence the conclusion.
\end{proof}

\subsection{Continuous approximation}\label{approx}

In this paragraph we formally derive differential equations corresponding to the PFP dynamic discrete equations (\ref{eq:ua}-\ref{eq:ub}) as presented in the previous section. Hopefully, this will lead us to get easily some qualitative insights on the solutions $u$ and $v$.\\
To do so, we write the Taylor expansion of order 1 at the point $(x,y)$ in system (\ref{eq:ua}-\ref{eq:ub}). After a quick computation, we get the following system of Partial Differential Equations (PDEs). Note that for the sake of simplicity we shorten the notations as follows: $sell $ becomes $s$, $buy$ becomes $b$, $Q^a$ becomes $x$ and $Q^b$ becomes $y$.
\begin{equation*}
\begin{split}
\mbox{(Ask) } \; \; \; \; 0 &= [(\lambda_bR_b^\ominus +\lambda^-)\frac 1 x (p^b(x)-u) - c_a] \\
& + [\lambda_s R^\oplus_s  - \lambda_bR^\ominus_b -\lambda^-] \cdot \partial_xu + [\lambda_b R^\oplus_b  - \lambda_s R^\ominus_s -\lambda^- ] \cdot \partial_yu ,
\end{split}
\end{equation*}
\begin{equation*}
\begin{split}
\mbox{(Bid) } \; \; \; \; 0 &= [(\lambda_s R_s^\ominus +\lambda^-)\frac 1 y( p^s(y)-v) + c_b] \\
& +  [\lambda_s R^\oplus_s  - \lambda_bR^\ominus_b -\lambda^-] \cdot \partial_xv + [\lambda_b R^\oplus_b  - \lambda_s R^\ominus_s -\lambda^- ] \cdot \partial_yv. 
\end{split}
\end{equation*}
Recall that $u, v, R_b, R_s$ are estimated at $(x,y)$ and $R_b$ depends upon $v$, resp. $R_s$ depends upon $u$. Consequently, $R_b$ and $R_s$ are the coupling terms in the PDE system (Ask)-(Bid).\\
The system has to be understood locally in the four regions 
\begin{equation*}
\begin{split}
&R^{++}\!\!=\{(x,y),  R_s^\oplus(x,y)=R_b^\oplus(x,y)=1\}, R^{--}\!\!=\{(x,y), R_s^\ominus(x,y)=R_b^\ominus(x,y)=1\},\\
&R^{+-}\!\!=\{(x,y),  R_s^\oplus(x,y)=R_b^\ominus(x,y)=1\}, R^{-+}\!\!=\{(x,y), R_s^\ominus(x,y)=R_b^\oplus(x,y)=1\}.\\
\end{split}
\end{equation*}
Now we can write the general form of the first order system of coupled PDEs. 
\begin{eqnarray}\label{eq:gensyst}
&&0 = \gamma_a(u,v,x,y)+\alpha (u,v,x,y) \partial_x u +\beta (u,v,x,y) \partial_yu\\
\label{eq:gensystt}
&&0 = \gamma_b(u,v,x,y)+ \alpha(u,v,x,y) \partial_x v + \beta(u,v,x,y) \partial_yv,
\end{eqnarray}
where $\gamma_a,\gamma_b$, $\alpha, \beta$ have some good symmetry properties to be described later on.

\paragraph{The MFG framework.}
The model is of course a Mean Field Game. As mentioned in section \ref{sec:MFG}, there are continuous entries and exits of players (modeled with Poisson processes). Therefore the basis assumptions are fulfilled: continuum of atomized and anonymous players.\\
Comparing equations (\ref{eq:gensyst})-(\ref{eq:gensystt}) and (\ref{eq:MFGsyst}), it is easy to notice that the equilibrium equations have the same form as the monotone system characterizing some MFG equilibria.

\paragraph{Second order terms.}We kept only the first order terms in the equations. The second order terms to be added to the equations are:\\
\begin{equation*}
\begin{split}
\!\! \mbox{\tiny{In (Ask)}}\;\frac{q^2}{2}\!\left[\frac 2 x (\lambda_bR_b^\ominus \!+\! \lambda\!^-\!)\partial_x u \!
+ \!\lambda\!^-\!  \Delta u \!+\! ( \lambda_sR^\oplus_s\!+\!\lambda_bR_b^\ominus )  \partial_{xx}u\!+ \!(\lambda_sR_s^\ominus\!+\!\lambda_bR^\oplus_b )  \partial_{yy}u\!\right]\!\!,
\end{split}
\end{equation*}
\begin{equation*}
\begin{split}
\!\! \mbox{\tiny{In (Bid)}}\;\frac{q^2}{2}\!\left[\frac 2 y ( \lambda_sR_s^\ominus\!+\! \lambda\!^-)\partial_y v \!
+ \!\lambda\!^-\!  \Delta v +\! ( \lambda_sR^\oplus_s\!+\! \lambda_b R_b^\ominus)  \partial_{xx}v\!+ \!(\lambda_s R_s^\ominus \!+\!\lambda_b R^\oplus_b )  \partial_{yy}v\!\right]\!\!.
\end{split}
\end{equation*}

\section{Equilibrium analysis}\label{sec:eqanal}

\subsection{Change of variables}

From now on we focus on the symmetric case where $\lambda_s = \lambda_b = \lambda$ and $c_a = c_b=c$.
First it is convenient to notice that in this important case, we have the following property:
\[\alpha = \beta = [\lambda (R_s^\oplus(u,x,y)  - R_b^\ominus(v,x,y)) -\lambda^-].\]
We will see later that this property allows to solve the problem thanks to the characteristics method.\\
\noindent
There is a very welcome change of variables that we will use throughout this section.
We define 
\begin{equation}\label{cov}
\tilde u = (u-P)/q \; \; \mbox{ and } \; \; \tilde v = (v-P)/q.
\end{equation}
\noindent
Then the (Ask)-(Bid) system reads
\begin{equation}\label{eq:pden}
\begin{split}
0 &= [(\lambda\tilde R_b^\ominus +\lambda^-)\frac 1 x (\frac {\delta}{x-q}-\tilde u) - \frac{c}{q}]  + [\lambda \tilde R^\oplus_s  - \lambda \tilde R^\ominus_b -\lambda^-] \cdot (\partial_x \tilde u + \partial_y \tilde u) ,\\
0 &= [(\lambda \tilde R_s^\ominus +\lambda^-)\frac 1 y( \frac {-\delta}{y-q}-\tilde v) + \frac{c}{q}] + [\lambda \tilde R^\oplus_s  - \lambda\tilde R^\ominus_b -\lambda^-] \cdot (\partial_x\tilde v+  \partial_y\tilde v). 
\end{split}
\end{equation}

\begin{prop}\label{prop:antis}
Assume that system (\ref{eq:pden}) admits a unique solution $(\tilde u, \tilde v)$, then it is antisymmetric, that is:
\[\forall (x,y), \;  \tilde v(x,y) = - \tilde u(y,x).\]
\end{prop}


\noindent
The general form of the system (\ref{eq:pden}) is as follows:
\begin{eqnarray}\label{eq:gensyst1}
&&0 = \;\;\gamma(\tilde u,\tilde v,x,y)+ \alpha (\tilde u,\tilde v,x,y) (\partial_x \tilde u+\partial_y \tilde u)\\
&&0 = \!-\gamma(\tilde v,\tilde u,y,x)+ \alpha(\tilde u,\tilde v,x,y) (\partial_x \tilde v+\partial_y \tilde v).
\end{eqnarray}

 \subsection{First Order Analysis}\label{sec:FOA}
 
Here we explore formally some aspects of the first order approximation to the solution.\\
The key point of the analysis is that in the two equations of system (\ref{eq:pden}), the derivative terms are the same, so that we conclude that the characteristics satisfy \[\dot{x} = \dot{y} = \alpha \Rightarrow x = y+ k.\]
Note that the reasoning of this paragraph holds on the region below the diagonal, but can be trivially extended to the whole domain by symmetry arguments.\\
We heuristically suppose that for a given $k$, and along the characteristic line $y = x - k$, there is a first point $M_0 = (x_0,y_0)$ where the sellers become Liquidity Consumer, that is $M_0$ is a point at the boundary of the regions $R^{++}$ and $R^{-+}$.\\
Then there is a second point $M_1 = (x_1,y_1)$, with $x_1 \geq x_0$ and $y_1 \geq y_0$ where the buyers become Liquidity Consumer, that is $M_1$ is a point at the boundary of the regions $R^{-+}$ and $R^{--}$.\\
First recall that:
\begin{equation*}
\begin{split}
& R^{++}  \mbox{ is defined by } R_s^\oplus = 1 \mbox{ and } R_b^\oplus = 1, \\
& R^{-+}  \mbox{ is defined by } R_s^\oplus = 0 \mbox{ and } R_b^\oplus = 1, \\
& R^{--}  \mbox{ is defined by } R_s^\oplus = 0 \mbox{ and } R_b^\oplus = 0. \\
\end{split}
\end{equation*}
We can write the differential equations on the three regions mentioned above:
\begin{equation}\label{eq:firstordereqs}
\begin{split}
(A_{R^{++}}) \;\;\;0 &= \Big[\frac {\lambda^-} x (\frac {\delta}{x-q}-\tilde u) - \frac{c}{q}\Big]  + [\lambda  -\lambda^-] \cdot (\partial_x \tilde u + \partial_y \tilde u) ,\\
(B_{R^{++}}) \;\;\;0 &= \Big[\frac {\lambda^-} y( \frac {-\delta}{y-q}-\tilde v) + \frac{c}{q}\Big] + [\lambda  -\lambda^-] \cdot (\partial_x\tilde v+  \partial_y\tilde v), \\
(A_{R^{-+}}) \;\;\;0 &= \Big[\frac {\lambda^-} x (\frac {\delta}{x-q}-\tilde u) - \frac{c}{q}\Big]  + [  -\lambda^-] \cdot (\partial_x \tilde u + \partial_y \tilde u) ,\\
(B_{R^{-+}}) \;\;\;0 &= \Big[\frac {\lambda + \lambda^-} y( \frac {-\delta}{y-q}-\tilde v) + \frac{c}{q}\Big] + [  -\lambda^-] \cdot (\partial_x\tilde v+  \partial_y\tilde v), \\
(A_{R^{--}}) \;\;\;0 &= \Big[\frac {\lambda + \lambda^-} x (\frac {\delta}{x-q}-\tilde u) - \frac{c}{q}\Big]  + [-\lambda  -\lambda^-] \cdot (\partial_x \tilde u + \partial_y \tilde u) ,\\
(B_{R^{--}}) \;\;\;0 &= \Big[\frac {\lambda + \lambda^-} y( \frac {-\delta}{y-q}-\tilde v) + \frac{c}{q}\Big] + [-\lambda  -\lambda^-] \cdot (\partial_x\tilde v+  \partial_y\tilde v).
\end{split}
\end{equation}
The equations are relatively simple in each region. The tricky point is, as always, to stick together the solutions of each region. First we compute the boundaries of the regions.

\paragraph{First order boundaries.}

Let's note $M_0$ the first order boundary between $R^{++}$ and $R^{-+}$ and $M_1$ between $R^{-+}$ and $R^{--}$.

\begin{prop}[First order boundary between $R^{++}$ and $R^{-+}$]
  \label{prop:M0}
  The diagonal point of the boundary $M_0$ is the point
  \begin{equation}\label{x0diag}
    (x_0^*,x_0^*) = (q+\sqrt{q^2+8/\eta})/2
  \end{equation}
  and the boundary $M_0$ is given by the set of points $(x_0,y_0)$ verifying:
  \begin{equation}
    \label{eq:M0}
    (x_0,y_0) = \Big(x_0,l(x_0):= q+ \Big(\eta x_0-\frac 1 {x_0-q}\Big)^{-1}\Big), \; \forall x_0 \geq x_0^*,
  \end{equation}
  where $\eta := {c}/{( \delta q\lambda^-)}$.
\end{prop}


\begin{prop}[First order boundary between $R^{-+}$ and $R^{--}$]
  \label{prop:M1}
  The boundary $M_1$ is defined by the set of points $(y_1+x_0-l(x_0),y_1), \; \forall x_0 \geq x_0^*,$
  where
  \begin{equation}\label{eq:boundM1}
    y_1 \mbox{ verifies } f_{x_0-l(x_0)}(y_1) = \frac \delta {y_1+x_0-l(x_0)-q}.
  \end{equation}
\end{prop}


See Sections \ref{sec:proofM0} and \ref{sec:proofM1} in Appendix for the proofs.
\\

Figure \ref{fig:FOS} exhibits an instance of the first order curves. We observe that near the diagonal, there is a region where several solutions could happen. The first order analysis thus shows the global form of the shape of the solution (since it is based on the curves $M_0, \; M_1$), and that considering higher order terms is necessary to understand what happens in the region near the diagonal.

\begin{figure}[!h]
\includegraphics[width=.8\linewidth]{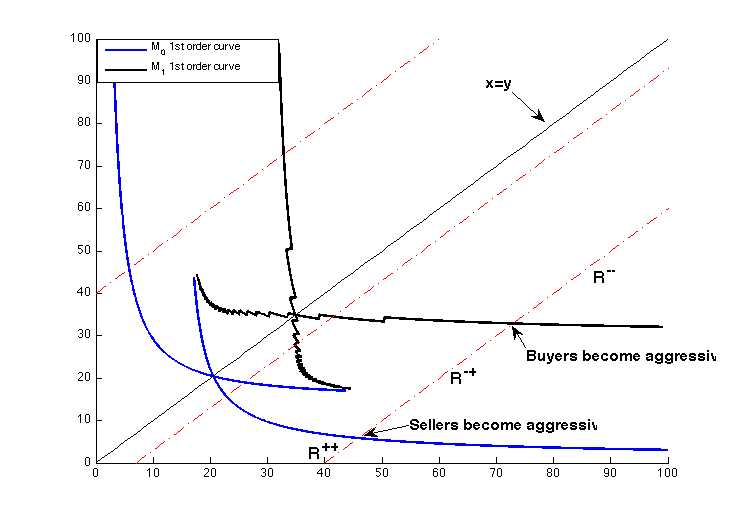}
\caption{First order decision curves}\label{smallq}
\label{fig:FOS}
\end{figure}

\subsection{Second Order Equations}
According to section \ref{approx}, the general form of the second order equations is:
\begin{equation}\label{eq:gensyst2}
\begin{split}
&0 = \;\;\gamma(\tilde u,\tilde v,x,y)+ \alpha (\tilde u,\tilde v,x,y)(\partial_x \tilde u+\partial_y \tilde u)\\
&\;\;\;\;\;+ q \Big( \rho (\tilde v,x,y))\partial_x \tilde u + \xi_1 (\tilde u,\tilde v,x,y) \partial_{xx} \tilde u+ \xi_2 (\tilde u,\tilde v,x,y) \partial_{yy} \tilde u \Big),\\
&0 = \!-\gamma(\tilde v,\tilde u,y,x)+ \alpha(\tilde u,\tilde v,x,y) (\partial_x \tilde v+\partial_y \tilde v)+\\
&\;\;\;\;\;+q \Big( \rho (\tilde u,y,x)) \partial_y \tilde v + \xi_1 (\tilde u,\tilde v,x,y) \partial_{xx} \tilde v+ \xi_2 (\tilde u,\tilde v,x,y) \partial_{yy} \tilde v \Big),
\end{split}
\end{equation}
where:\\
 $\rho = \frac 1 x (\lambda R_b^\ominus+\lambda^-)$,  $\xi_1 = (\lambda(R^\oplus_s\!+\!R_b^\ominus)+\lambda^-)/2$, and 
$\xi_2 = (\lambda(R^\ominus_s\!+\!R_b^\oplus)+\lambda^-)/2.$

\bigskip
\noindent
See Section \ref{sec:regions2nd} in Appendix for the local equations on the same four regions.

\bigskip
\noindent
In the next part, we provide several example of markets based on the model.

\section{Applications}\label{sec:appli}

This section is dedicated to applications of the MFG model to study the outcome of a combination of different trading behaviors in the same order book.

The purpose here is not to study how the price discovery operates on the long term, but how microstructure effects can deviate transaction prices from the fair price.
This model will thus explain how the state of the liquidity can change the dynamics of the price while forming an equilibrium price. This equilibrium can potentially deviate from the latent (or fair) price. The main drivers of these modifications will be the behavior of trading agents, and specifically the average size of their orders, their speed, their waiting cost, and how often they use optimized strategies (see Table \ref{tab:paramIIHFT} for a qualitative description of the main parameters of their strategies).

In this section, we will investigate theoretically and using simulations the reasons why the price deviates or not from the ``\emph{fair price}'' (exogenously fixed).
The variables of interest are:
\begin{itemize}
\item the asymptotic state of the liquidity offer (i.e. the size of the bid queue and ask queue): are they large or small? are they balanced?
\item The average transaction price: how far away it is from the fair price? 
\item The average value of the bid-ask spread; in which conditions is it high or low?
\end{itemize}
%
Having in mind that each time an agent buys or sells she suffers from market impact, i.e. consuming liquidity implies paying enough to find counterparts (this premium decreases with the size of the consumed queue), the strategy of each market participant affects her price. We will thus be able to compute an average price for each class of market participants, answering the question: do the institutional investors pay more than high frequency traders?

It will also allow us to compute an \emph{effective bid-ask spread} being twice the difference between the mid price and the transaction price; it will not be the same for each market participant.

\begin{defi}[Effective bid-ask spread]
  The effective bid ask spread of an agent $A$ is the expected transaction price of its liquidity removing buying orders minus the one of its liquidity removing selling orders:
  \begin{equation}
    \label{eq:defesp}
    \psi^e(A) := \delta\cdot\Esp \left(  \left. \frac{q}{Q^a_t} \right| R^\ominus_{\rm buy} (A)\right)  + %
    \delta\cdot\Esp \left(  \left. \frac{q}{Q^b_t} \right| R^\ominus_{\rm sell} (A)\right)  .
  \end{equation}
\end{defi}
\noindent
The effective spread is higher for an impatient agent if the spread is larger when she consumes liquidity than when she provides liquidity.


Another important tool is the \emph{invariant measure} describing the repartition of the agents in the (liquidity) state space, being the probability of having the system in a specific region of the state space. Our state space is captured by the sizes of the two queues (the bid queue and the ask queue).\\
Thanks to the results obtained in the previous sections, we will be able not only to observe discrepancies between agents' behaviour and their outcome, but also to explain and understand them in details.

{\def\dmt{$\cdot 10^{-3}$} \def\dmd{$10^{-2}$}
\begin{table}[h!]
  \centering
  \begin{tabular}{c||c|c|c|c||c|c|}
                      & Test 1 & Test 2 & Test 3 & Test 4 & Test 5 & Test 6 \\ \hline
    $q_{ii}$            &    1      &  0.25   &    1     &     1      &    1     &     1 \\
    $\lambda_{ii}$ &    1      &     1     &    1     &   0.5     &    0.5  &    0.6 \\
    $\lambda^-_{ii}$ & 0.2   &   0.2   &    0.2   &   0.5    &  0.5    &   0.4 \\ 
    $c_{ii}\cdot q_{ii}$  &  2.5\dmt & 2.5\dmt & \dmd & 2.5\dmt & 2.5\dmt & 2.5\dmt \\\hline 
    $q_{\mbox{\tiny HFT}}$                & -   &  -   &    -     &     -          &    0.25    &   0.25 \\
    $\lambda_{\mbox{\tiny HFT}}$     & -    &     -     &    -     &   -     &    4  &     3.6 \\
    $\lambda^-_{\mbox{\tiny HFT}}$ & - &   -   &    -   &   -              &  0    &   0.4 \\ 
    $c_{\mbox{\tiny HFT}}\cdot q_{\mbox{\tiny HFT}}$ & -     & -       & - & -               & \dmd & \dmd \\\hline
  \end{tabular}
  \caption{Parameters defining the studied models.}
  \label{tab:models}
\end{table}
}

\noindent
A first subsection is dedicated to applications with models including one class of agents only, to understand and explain in details the mechanisms that our MFG model can render. 
In a second subsection we will use an \emph{heterogenous agent model}, allowing to understand the result of putting together more than one class of market participants. Here we mix Institutional Investors and High Frequency Traders.
Section \ref{sec:twogroupseq} presents a theoretical expansion of Section \ref{sec:obmodel} needed to handle more than one agent class.
Table \ref{tab:models} summarizes the different models and their parameters.


\subsection{Markets with Institutional Investors only} 
\label{sec:IIonlysims} 

\subsubsection{Modeling Institutional Investors}

Since we just want to model one class of market participants, their specification is not very important. It will become crucial when we mix different types of agents: the relative speed, the relative sizes of orders, etc., will play a role of paramount importance in the multi-agent simulations.

With one type of investors only we mainly focus on using realistic values and exploring the sensitivities of the emerging dynamics to the values of the parameters.
Note first that some parameters define the framework of the simulation and not the market participants themselves:
\begin{itemize}
\item we have seen this in the change of variable (\ref{cov}) that the \emph{fair price} $P$ has no impact on the dynamics, it is taken as a constant,
\item the \emph{market depth} $\delta$, playing a role in the
  expression of the market impact of one trade (at the first order it
  is homogenous to Kyle's lambda). Looking carefully at the market
  impact expression \eqref{eq:execprice}, it can be read that $\delta$
  is homogenous to the inverse of a quantity: dividing $\delta$ by two and multiplying quantities by two will not change the dynamics but relatively increase the waiting costs (that are proportional to the order size $q$).
\end{itemize}
\medskip

Other parameters are directly associated with the agent:
\begin{itemize}
\item the size of her orders $q$,
\item the intensity $\lambda$ of the Poisson process governing the arrival rate of smart routed orders;
\item the intensity $\lambda^-$ of the Poisson process governing the
  arrival of not smart routed orders (i.e. blindly sending market
  orders or having infinite waiting costs);
\item the cost of waiting per share $c$: waiting $dt$ seconds is worth $cq\,dt$.
\end{itemize}
\medskip

\noindent
Some simple statistics on equity markets can give reasonable figures for these parameters (
see \cite{citeulike:12047995} for more details about evolution of trading behaviours from 2007 to 2013):
\begin{itemize}
\item the intensity $\Lambda = \lambda+\lambda^-$ can be roughly estimated by the average number of trades per time unit ;
\item the size $q$ has no unit (it will have a role when compared to the size of HFT orders); for the sake of simplicity we will take it equals to one. For information the table gives the \emph{average trade size} and the \emph{average size at first limit}.
\end{itemize}

\subsubsection{Simulations and results}


First we consider the case of a single group of traders all with the same order size $q$.
The elementary algorithm we use to compute the equilibrium is as follow:
\begin{enumerate}
\item Initialize $u^0$ and $v^0$ (e.g. to the constant function equal to $P$)
\item Step $k$: \\
compute $u^{k}$ the solution to equation (\ref{eq:ua}) using the inputs $u^{k-1}$ and $v^{k-1}$,\\
compute $v^{k}$ the solution to equation (\ref{eq:ub}) using the inputs $u^{k}$ and $v^{k-1}$.
\end{enumerate}

\paragraph{Equilibrium as an invariant measure.}
Equilibrium visualization is made of the level sets of the stationary measure of queue sizes. The previous measure is computed from the transition probability depending upon $u$ and $v$. \\
More precisely, the transition process at a certain state $(Q_a,Q_b)$:
$$\begin{array}{ccccl}
(Q_a,Q_b) &\rightarrow & (Q_a,Q_b) &\mbox{ with probability } &1 - 2 \lambda dt - 2 \lambda^- dt\\
(Q_a,Q_b) &\rightarrow &(Q_a+q,Q_b) &\mbox{ with probability }& \lambda R_s^\oplus dt \\
(Q_a,Q_b) &\rightarrow &(Q_a-q,Q_b) &\mbox{ with probability } & \lambda R_b^\ominus dt + \lambda^- dt \\
(Q_a,Q_b) &\rightarrow &(Q_a,Q_b+q) &\mbox{ with probability } & \lambda R_b^\oplus dt\\
(Q_a,Q_b) &\rightarrow &(Q_a,Q_b-q) &\mbox{ with probability } &\lambda R_s^\ominus dt + \lambda^- dt, \\
\end{array}$$
where we use the probability of occurrence of events as described in equations \eqref{eq:ua} and \eqref{eq:ub}.

The resulting process has of course the Markov property.

\begin{figure}[!ht]\centering
\hspace{-.25cm}\subfigure[Mapping of the decision regions]{\includegraphics[width=0.75\linewidth]{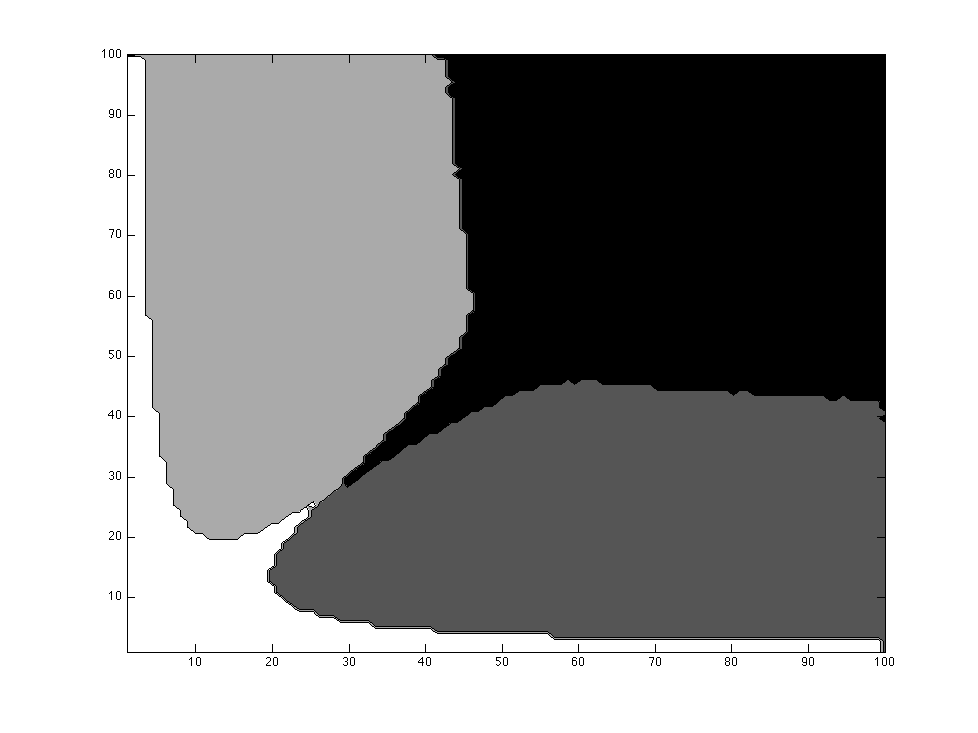}}
\hspace{-.5cm} \subfigure[Corresponding invariant measure]{\includegraphics[width=0.75\linewidth]{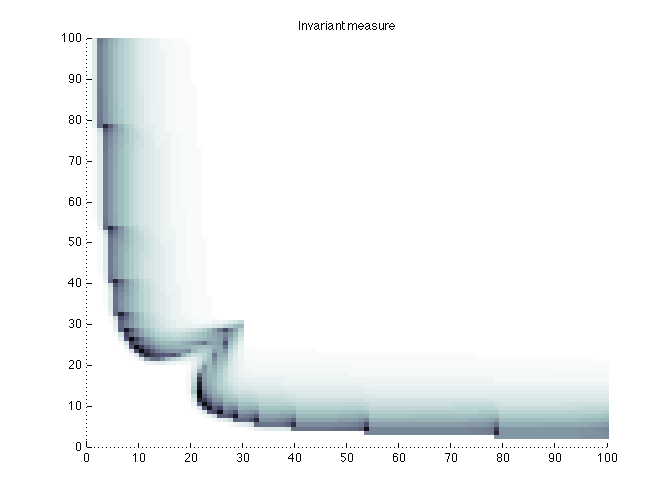}}
\caption{Test 1: the numerical solution for a single homogeneous specie of traders. %
(a) The routing decision regions have the expected form: for small queue sizes (in white) the buyers and sellers act mostly as liquidity providers, for large queues (in dark) they both act as liquidity consumers, in between (in grey) only one type of agent (buyers or sellers) consume liquidity while the other provides liquidity (this last case correspond to liquidity imbalances). %
(b) The invariant measure exhibits two symmetric cavities: (the white zones figure low concentration of agents while the dark ones are for frequent stable points of the state space); it reads that the liquidity imbalances ((a) grey zones) can be stable.%
}\label{fig:res1specie}
\end{figure}

\begin{figure}[!ht]
\includegraphics[width=.9\linewidth]{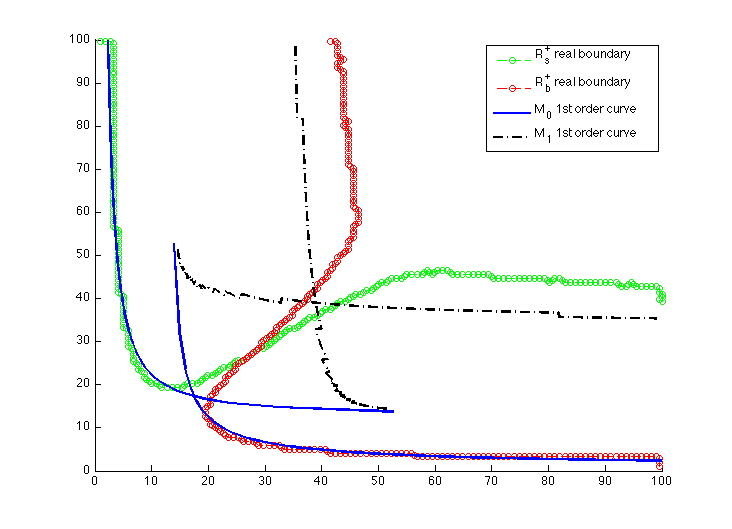}
\caption{Test 1: the numerically computed switching curves (red and green dots) tries to conciliate the curves analytically computed at order 1 (dotted and solid lines).}\label{fig:bounds1}
\end{figure}

\begin{figure}[!ht]
\includegraphics[width=\linewidth]{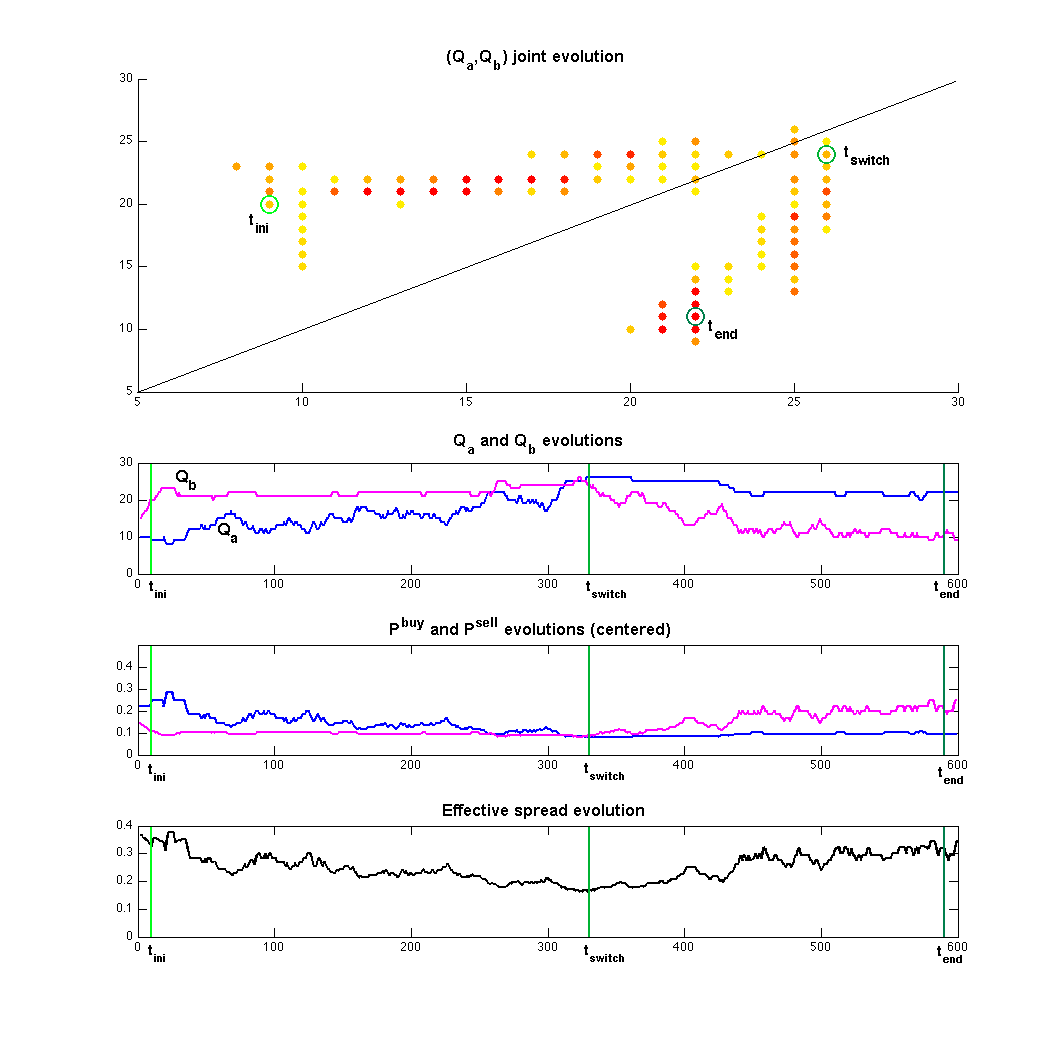} 
\caption{Here we show a particular simulation of the orderbook. We plot the evolution of various quantities for 600 instants. A sample trajectory of coupled queue sizes is plotted. Dots are colored from yellow (1 visit) to red (about 15 visits). Note that 3 milestones are introduced. We observe a change of regime at instant $t_{switch}$ where the market activity switches from the ask queue to the bid queue.}\label{fig:pd}
\end{figure}


\paragraph{Test 1: institutional traders with many SOR arrivals (i.e. very few impatient traders).}

The first numerical test corresponds to the following set of parameters: $q=1, \delta=2, c = 2,5 \times 10^{-3}, \lambda =1,$ and $ \lambda^-=0.2.$

We show the results in Figure \ref{fig:res1specie}. We observe that the decision regions $R^{++}, R^{-+}, R^{+-},$ and $R^{--}$ have the expected form. We also remark that the second order term selects a particular solution amongst all order one solutions. Mathematically, this has to be linked to the notion of viscosity solutions, but we do not enter in the details here \cite{citeulike:7800477}. 

Below the diagonal, that is for values of $Q_b$ smaller than $Q_a$, the region where both sellers and buyers are Liquidity Provider corresponds to small $Q_a$ and $Q_b$, then the sellers turn to be Liquidity Consumer while the buyers remain liquidity adders, and finally they also turn to be Liquidity Consumer.

\bigskip
\noindent
The invariant measure is almost concentrated on the points where both sellers and buyers turn to be Liquidity Consumer, i.e. it is concentrated on the boundary curve $M_0$ describing the frontier between $R^{++}$ and $R^{-+} \& \; R^{+-}.$ From now on, we refer to this curve as the {\it \PC{} switching curve} for Provider to Consumer switching curve. We symmetrically define the {\it \CP{} switching curve} as the frontier between $R^{+-} \; \& \; R^{-+}$ and $R^{--}$.

However, the most remarkable point is that the invariant measure shows two bumps, located in the cavities of the \PC{} switching curve. In the new reference frame after a $\pi/4$ axis rotation, the cavities corresponds to the global minimum points of the \PC{} switching curve.

Here the economic intuition is that there are two symmetric liquidity pools, one on the buy side, where only buy orders are completed, and conversely for the sell side.

\begin{figure}[!h] 
\subfigure[The invariant measure has the same form, it is concentrated on smaller values.]{\includegraphics[width=.85\linewidth]{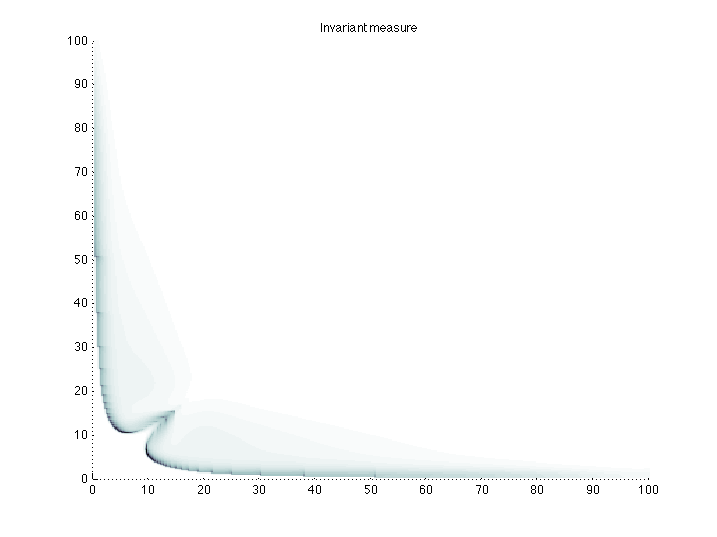}} \subfigure[The real PC switching curves is closer to the 1st order switching curve. ]{\includegraphics[width=.85\linewidth]{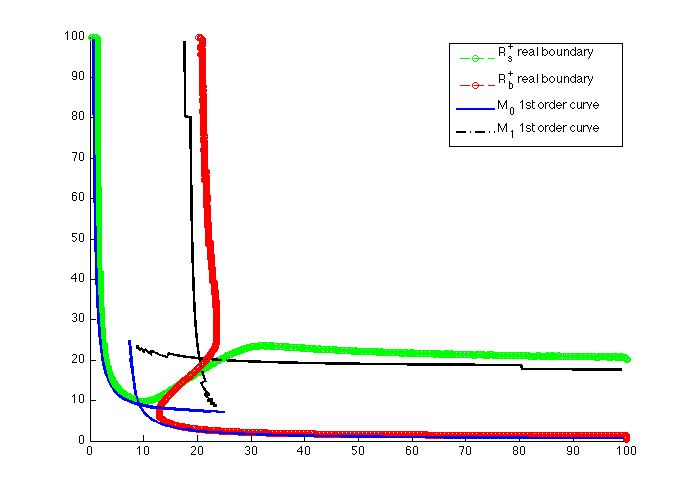}}
\caption{Test 2 (smaller order size $q$): comparing the invariant measure (a) and the switching curves (b) with the one of test 1 (Figure \ref{fig:res1specie}:b and Figure \ref {fig:bounds1}), we deduce the second order terms impacts.}\label{smallq}
\end{figure}

\bigskip 
\noindent 
Figure \ref{fig:bounds1} shows the form of the routing decision boundaries {\it versus} the first order analytical curves derived earlier. It is noteworthy to observe that the {\it real} switching curves tries to conciliate the curves analytically computed at order 1. However, it approximates better the \PC{} switching curve than the \CP{} switching curve. 

\bigskip
\noindent
Several visualizations of system trajectories are possible. We display the time evolution (for 600 instants) of various quantities in Figure \ref{fig:pd}. In this example there are mainly two distinct regimes: from instant $t_{ini}$ to $t_{switch} = 330$, the activity is mainly concentrated on the ask queue. During the second period, most of the activity holds on the bid side.

The coupled trajectory of queue sizes in the space $Q_a \times Q_b$ is another possible visualization.  In the plot, dots are colored from yellow to red, according to the number of time the queue system passes through the corresponding size configuration. Here also we see the hange of regime at $t_{switch}$ where the process  goes through the diagonal, jumping from the ask activity zone to the bid activity one. One more time both liquidity configurations are visible. Above the diagonal $Q_a=Q_b$, red dots are more likely to be horizontally distributed (meaning that most of transactions hold on the ask side), and symmetrically below the diagonal. 

Both visualizations confirm the phenomenon that one could expect after looking at the invariant measure plot.

\begin{figure}[h!]\centering
\hspace{-.25cm}%
\subfigure[Test 1: $c = 2.5\times10^{-3}$]{\includegraphics[width=.5\linewidth]{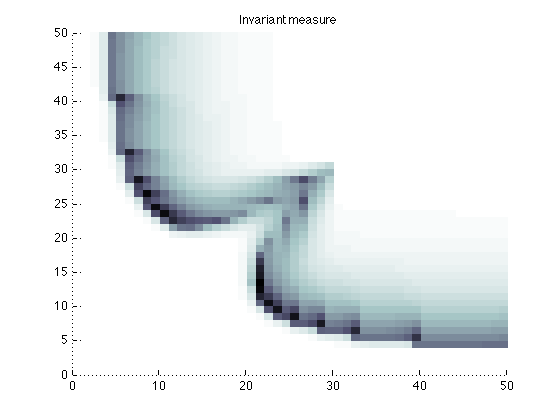}} %
\subfigure[Test 3: $c_h = 10^{-2}$]{\includegraphics[width=.5\linewidth]{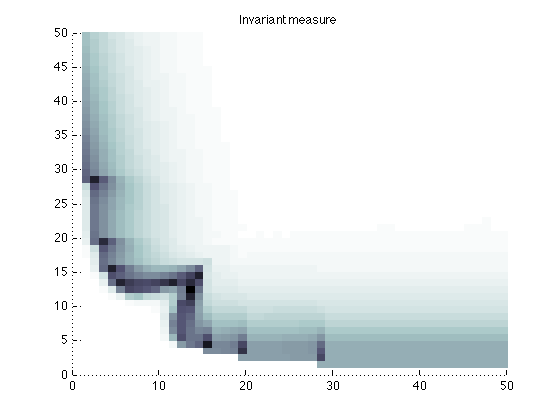}} 
\caption{The two symmetric bumps vanish and a single bump appears on the diagonal}\label{fig:cC}
\end{figure}

\begin{figure}[h!]
\hspace{1cm}\includegraphics[width=.75\linewidth]{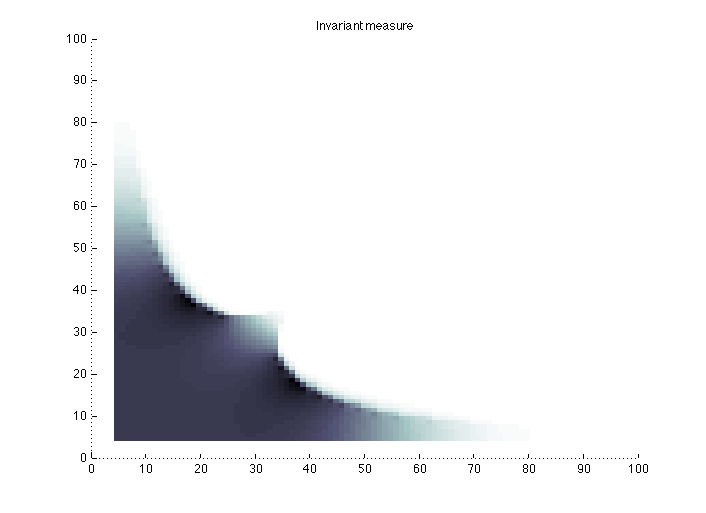}
\caption{Test 4: the case with half of very impatient (i.e. Non-SOR) orders.}\label{fig:halfhalf}
\end{figure}

\paragraph{Test 2: order size impact.} 
Now we only change the value of $q$, and take it smaller than in test 1: $q=0,25$. Figure \ref{smallq} shows that the real \CP{} switching curve is closer to the 1st order curve, which is natural since taking $q$ smaller means that the second order term impacts are shrunk (as expected comparing equations \ref{eq:secondordereqs} with equations \ref{eq:firstordereqs}).

\paragraph{Test 3: risk aversion impact.}
Figure \ref{fig:cC} shows the impact of a bigger $c$ on the solution. We compare the results obtained for the set of parameters of Test 1 with the results obtained for the same parameters except the value of the new risk aversion $4 \times c = 10^{-2}$. We observe mainly two effects:
\begin{itemize}
\item queue sizes are shrunk (from about 30 to 10);
\item the invariant measure maximum is now on the diagonal. Therefore, the two antisymmetric liquidity pools progressively disappear.
\end{itemize}

\paragraph{Test 4: the case with half non-SOR (i.e. impatient) orders.}
Here we provide a stationary equilibrium instance in a case where half of the order arrivals are non-SOR. Figure \ref{fig:halfhalf} shows that at the equilibrium there are still two symmetric regions of concentration at the neighborhood of the \PC{} switching curves. In such a case we observe more density on low queue sizes.

\subsubsection{Possible liquidity imbalance with one class of participant only}

First remind that the \emph{type of a participant} is described by the way she interacts with order books. 
Hence a pension fund taking long term positions, a low frequency statistical arbitrageur, and the hedging desk of an investment bank will have the same \emph{type}. The important elements being they all:
\begin{itemize}
\item take a decision before starting to interact with the order book,
\item do not use a smart order router systematically (i.e. can be very  impatient),
\item trade with relatively large orders, even once their \emph{meta orders} have been split thanks to an optimal trading scheme (like in \cite{citeulike:5797837},\cite{OPTEXECAC00},\cite{citeulike:6615020} or \cite{GLFT}).
\end{itemize}
\medskip
\noindent
The outcome of the application to one class of investors is that the market can suffer for long and stable liquidity imbalances. We have seen that in such typical cases the bid and ask queues are in an asymmetric configuration:
\begin{itemize}
\item one of the queues (the ask one, for instance) is significantly shorter than the other,
\item the flow of buyers considers that the price to pay to wait is too high and accept to pay the market impact on a small queue,
\item the flow of seller notices that they can obtain a fast trade being passive (i.e. going into the bid queue), since 100\% of the buyers are now impatient.
\end{itemize}
This leads to a stable state of the order book: the invariant measure sees two symmetrical concentrations of such configurations, dominating more balanced states located in the diagonal (see Figure \ref{fig:res1specie}:b, \ref{smallq}:a and \ref{fig:halfhalf}).

\bigskip
\noindent
In such a situation, we can say that \emph{liquidity calls for liquidity}: the conjunction of a high rate of consuming orders at the smaller queue and of an high arrival rate of liquidity on the same queue feeds an equilibrium.

\bigskip
\noindent
During such a configuration the transaction price is significantly different from the exogenous \emph{fair price}. Since the model is stochastic its state will nevertheless evolve to explore other configuration (see Figure \ref{fig:pd} for a trajectory instance). Nevertheless the form of the invariant measure indicates that the fraction of time during which the model is in such inefficient configurations dominates.

\bigskip
\noindent
When the behaviors of participants are so similar that they create liquidity imbalances, it is often proposed to add a population of market makers (see \cite{ho1983dynamics}), hoping that it will break these ``liquidity circles'' and bring back the invariant measure on the diagonal.
One argues that high frequency traders are a modern version of agents of this kind \cite{citeulike:8423311}. The goal of the next subsection is to study their influence on the invariant measure inside our MFG modelling framework.


\subsection{Introducing High Frequency Traders}

\subsubsection{Equations and quantities for two groups} 
\label{sec:twogroupseq}

The model is such that considering several types of traders is not a big deal. This is good news since our aim is to get insights on the role of High Frequency Traders in the scope of our model.\\
We therefore split the agents into two subsets:
  \begin{itemize}
  \item Institutional Investors, with a smaller intensity $\lambda_1$, but with bigger sizes $q_1 $ and risk aversion $c_1$.
  \item HFTs, with a higher intensity $\lambda_2$ and smaller sizes $q_2$ and risk aversion $c_2$,
  \end{itemize}
  The two groups also differentiate by having specific $\lambda_1^-, \lambda_2 ^-$ (i.e. impatient flows). 
    This leads to twice the value functions we had.\\
We also have to consider HFT's routing decisions $R_{buy}^\oplus(v_2 , x, y+q_2 ):= \one _{v_2 (x,y+q)<p_{q_2}^{buy}(x)}$ (symmetrically $R_{sell}^\oplus)$.\\
Now we write the equation of sellers' value functions $u_1, u_2$. As before, buyers' value function equations can be easily derived by simple symmetry arguments.
\begin{equation}\label{eq:2gps}
    \begin{split}
  \hspace*{.35cm}&    k \cdot u_1(x,y) = \\
  &  \;\;\;\;(\lambda_1 \RS^\ominus(u_1,x+q_1,y)+\lambda_1^-) u_1(x,y-q_1) + \lambda_1 \RS^\oplus(u_1,x+q_1,y) u_1(x+q_1,y) \\
  & + (\lambda_1 \RB^\ominus(v_1,x,y+q_1) + \lambda_1^-) [\frac {q_1} x \PB_1(x)+(1-\frac {q_1} x)u_1(x-q_1,y)] \\
  & + (\lambda_1 \RB^\oplus(v_1,x,y+q_1) u_1(x,y+q_1)) + (\lambda_2 R _{buy}^\oplus(v_2 ,x,y+q_2 ) u_1(x,y+q_2 )\\
  &  +(\lambda_2 R _{sell}^\ominus(u_2 ,x+q_2 ,y)+\lambda_2^-) u_1(x,y-q_2 ) + \lambda_2 R _{sell}^\oplus(u_2 ,x+q_2 ,y) u_1(x+q_2 ,y) \\
  & + (\lambda_2 R _{buy}^\ominus(v_2 ,x,y+q_2 ) + \lambda_2^-) [\frac {q_2 } x  p_2 ^{buy}(x)+(1-\frac {q_2 } x)u_1(x-q_2 ,y)] \\
  & -c_1q_1) ,
  \end{split}
\end{equation}
\begin{equation}\label{eq:2gps_B}
    \begin{split}
  \hspace*{.35cm}&    k \cdot u_2(x,y) = \\
  &  \;\;\;\;(\lambda_1 \RS^\ominus(u_1,x+q_1,y)+\lambda_1^-) u_2(x,y-q_1) + \lambda_1 \RS^\oplus(u_1,x+q_1,y) u_2(x+q_1,y) \\
  & + (\lambda_1 \RB^\ominus(v_1,x,y+q_1) + \lambda_1^-) [\frac {q_1} x \PB_1(x)+(1-\frac {q_1} x)u_2(x-q_1,y)] \\
  & + (\lambda_1 \RB^\oplus(v_1,x,y+q_1) u_2(x,y+q_1)) + (\lambda_2 R _{buy}^\oplus(v_2 ,x,y+q_2 ) u_2(x,y+q_2 ))\\
  &  +(\lambda_2 R _{sell}^\ominus(u_2 ,x+q_2 ,y)+\lambda_2^-) u_2(x,y-q_2 ) + \lambda_2 R _{sell}^\oplus(u_2 ,x+q_2 ,y) u_2(x+q_2 ,y) \\
  & + (\lambda_2 R _{buy}^\ominus(v_2 ,x,y+q_2 ) + \lambda_2^-) [\frac {q_2 } x  p_2 ^{buy}(x)+(1-\frac {q_2 } x)u_2(x-q_2 ,y)] \\
  & - c_2q_2 ,
  \end{split}
\end{equation}
where $k = 2 (\lambda_1+ \lambda_2 + \lambda_1^- + \lambda_2^-)$.\\
At this stage, it is important to remark that the only difference in equations (\ref{eq:2gps}) and (\ref{eq:2gps_B}) is the term $c_1q_1$. As a consequence, $u_1$ and $u_2$ coincide as soon as $c_1q_1 = c_2q_2$, which is the reference case we study in the present work.

We take $c_1 q_1=c_2 q_2$ by purpose: thanks to this choice HFT will not have an exogenous advantage coming from a lower waiting costs. They will have to exploit their size and frequency specificities (they send more orders, of smaller size).

Note that in this case, a first order solution can be explicitly calculated using the methodology of the previous section.

\bigskip
\noindent
In this section we will compare equilibria in terms of average transaction prices and spread. Let us detail the way we define average prices.
Before writing down the average price equations, we need to introduce some notations. For the sake of simplicity we only work with sellers. Buyers notations and equations can be easily derived by symmetry.
\begin{equation*}
\begin{split}
& \mbox{Empirical stationary measure: } \hat m(Q_a,Q_b) \\
&  \mbox{Type's $i$ stationary proportion: } \hat \gamma_i(Q_a,Q_b) 
\end{split}
\end{equation*}

  \begin{figure}[h!]
  \begin{center}
  \includegraphics[width=.85\linewidth,height=.65\linewidth]{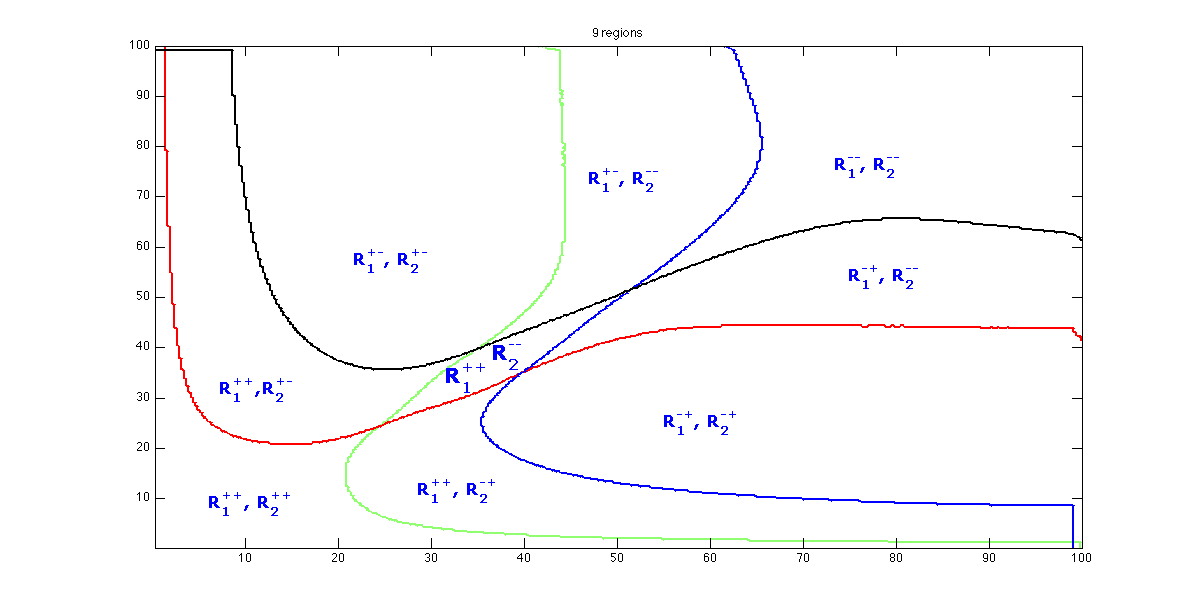}
\end{center} \caption{There are nine regions in terms of trader type (HFT\emph{vs} Institutional Investor) and trader action (LP{} \label{fig:9reg} \emph{vs} LC{})}
\end{figure}
\noindent
As previously, there are several regions defined by the LP{} or LC{} behavior of traders. Figure \ref{fig:9reg} shows an instance with 9 regions. Two cases may happen, depending upon the relative size of $\lambda_i^-,\lambda_i, i=1,2$. We do not want to enter the (technical) details nor review all possible cases, but we would like to mention that a necessary condition for the existence of several regions is that $\sum \lambda_i > \sum \lambda^-_i$, that is there is globally more SOR than non-SOR traders in the system. In the first case, Institutional Investors switch first from LP{} to LC{}. In the second case, HFT switch first (which is the case in Figure \ref{fig:9reg}).
Note that the proportion ${\hat\gamma}_i$ of traders of type $i$ is constant in each region defined in Figure \ref{fig:9reg}. 
Consequently we can define the marginal distribution  of traders of type $i$ as:
\[ \hat m_i(Q_a,Q_b):= \hat \gamma_i(Q_a,Q_b) \hat m(Q_a,Q_b).\]
Now, in each region, there is a certain traded quantity $\xi$. In table 3 we provide the corresponding values of $\xi$ and $\gamma_i$.

\begin{table}[h!]
\begin{center}
\begin{tabular}{|c|c|c|}  \hline &  $\gamma_i$ & $\xi$  \\ \hline
$\mathcal{R}_1\! := \! (R_1^{++},R_2^{++})$ & $\frac{\lambda_i}{\lambda_1+\lambda_2}$ & $ (\lambda_1^- q_1/Q_a, \lambda_2^- q_2/Q_a)$ \\ \hline 
$\mathcal{R}_2\! := \! (R_1^{++},R_2^{--})$ & $\one_{i=1}$ & 
$(\lambda_1^-q_1/Q_a,\Lambda_2 q_2/Q_a)$ \\ \hline
$\mathcal{R}_3\! := \! (R_1^{--},R_2^{--})$ & 0 & 
- \\ \hline
$\mathcal{R}_4\! := \! (R_1^{++},R_2^{-+})$ & $\one_{i=1}$ & 
$ (\lambda_1^- q_1/Q_a, \lambda_2^- q_2/Q_a)$ \\ \hline
$\mathcal{R}_5\! := \! (R_1^{-+},R_2^{+-})$ & 0 & - \\ \hline
$\mathcal{R}_6\! := \! (R_1^{-+},R_2^{--})$ & 0 & - \\ \hline
$\mathcal{R}_7\! := \! (R_1^{++},R_2^{+-})$ & $\frac{\lambda_i}{\lambda_1+\lambda_2}$ & $(\lambda_1^-q_1/Q_a,\Lambda_2 q_2/Q_a)$ \\ \hline
$\mathcal{R}_8\! := \! (R_1^{+-}R_2^{+-})$ & $\frac{\lambda_i}{\lambda_1+\lambda_2}$ & $ (\Lambda_1 q_1/Q_a, \Lambda_2 q_2/Q_a)$ \\ \hline
$\mathcal{R}_9\! := \! (R_1^{+-},R_2^{--})$ & $\one_{i=1}$ & 
$ (\Lambda_1 q_1/Q_a, \Lambda_2 q_2/Q_a)$ \\ \hline
\end{tabular}
\end{center}\caption{Values of various quantities in each of the 9 regions}
\end{table}

\noindent
The general formula of the average prices are:

\bigskip
\noindent
$\bullet$ Type's $i$ LC{} proportion: 
\[M_{s,i}^- := \!\!\! \int_{(Q_a,Q_b) \in R^{-+}_{i} \bigcup R^{--}_{i}} \!\!\!\!\!\!\!\!\!\!\!\!\!\!\!\! \Lambda_i q_i  d\hat m(Q_a,Q_b)+ \!\!\! \int_{(Q_a,Q_b) \in R^{++}_{i} \bigcup R^{+-}_{i}}  \!\!\!\!\!\!\!\!\!\!\!\!\!\!\!\!\! \lambda^-_i q_i  d\hat m(Q_a,Q_b) \]

\noindent
$\bullet$ Type's $i$ LP{} proportion:  
\[ M_{s,i}^+ := \sum_{i=1}^{9}\int_{(Q_a,Q_b) \in \mathcal{R}_i}   \langle \xi(Q_a,Q_b), (1,1) \rangle\; d\hat m_i(Q_a,Q_b) \]

\noindent
$\bullet$ Type's $i$ price for Liquidity Consumer traders:
\[\bar p_{s,i}^- := \Big(\!\! \int_{(Q_a,Q_b) \in R^{-+}_{i} \bigcup R^{--}_{i}} \!\!\!\!\!\!\!\!\!\!\!\!\!\!\!\!\!\!\!\!\!\!\!\!\!\!\!\!\!\!\!\!\!\!\!\!\!\!\!\!\!\! \Lambda_i q_i \PS_{q_i}(Q_b) d\hat m(Q_a,Q_b)+ \!\! \int_{(Q_a,Q_b) \in R^{++}_{i} \bigcup R^{+-}_{i}}  \!\!\!\!\!\!\!\!\!\!\!\!\!\!\!\!\!\!\!\!\!\!\!\!\!\!\!\!\!\!\!\!\!\!\!\!\!\!\!\!\!\!\lambda^-_i q_i \PS_{q_i}(Q_b) d\hat m(Q_a,Q_b)\Big) / M_{s,i}^-\]
\noindent
$\bullet$ Type's $i$ price for Liquidity Provider traders:
\[\bar p_{s,i}^+ := \Big( \sum_{i=1}^{9}\int_{(Q_a,Q_b) \in \mathcal{R}_i}  \!\!\!\!\!\!\!\!\!\!\!\!\!\!\!\!\!\!\!\!\!\!\!\!  \langle \xi(Q_a,Q_b), (\PB_{q_1}(Q_a), \PB_{q_2}(Q_a))\rangle\; d\hat m_i(Q_a,Q_b) \Big)/M_{s,i}^+\]

\noindent
Finally the average price for sellers can be simply deduced:
\[\bar p_{s,i} = \frac{\bar p_{s,i}^- M_{s,i}^- + \bar p_{s,i}^+ M_{s,i}^+}{ M_{s,i}^- + M_{s,i}^+}.\]
  
\subsubsection{ Numerical tests}
 
\begin{figure}[!h]
\subfigure[Test 4: Institutional Investors only]{\includegraphics[width=.8\linewidth]{II_only}} %
\subfigure[Test 5: Institutional Investors and HFT]{\includegraphics[width=.8\linewidth]{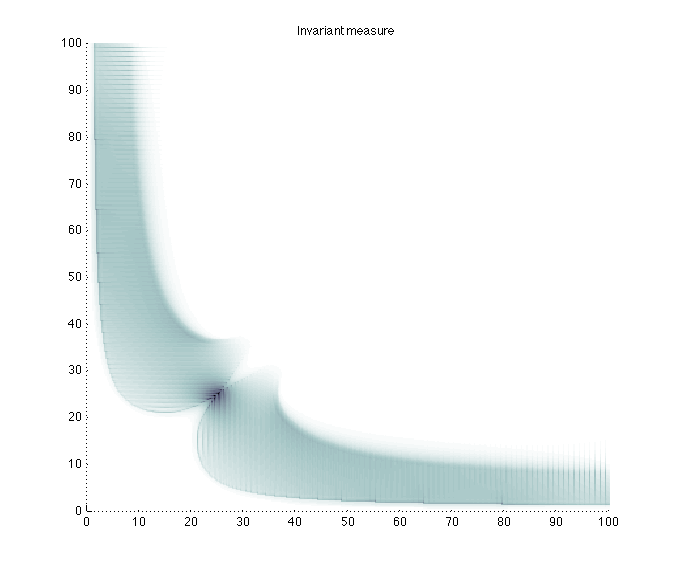}} %
\caption{Comparison of invariant measures:  effect of adding HFT (tests 3 and 5).}
\label{fig:HFTII}
\end{figure}

\begin{figure}[!ht]
\subfigure[Test 1: Institutional Investors only]{\includegraphics[width=.8\linewidth]{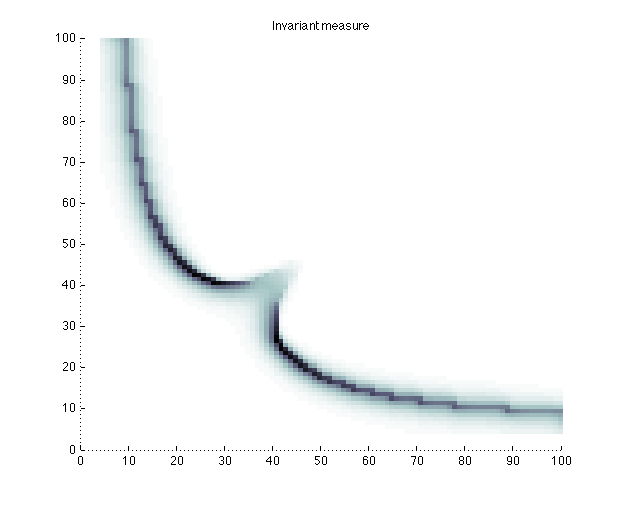}} %
\subfigure[Test 6: Institutional Investors and HFT]{\includegraphics[width=.8\linewidth]{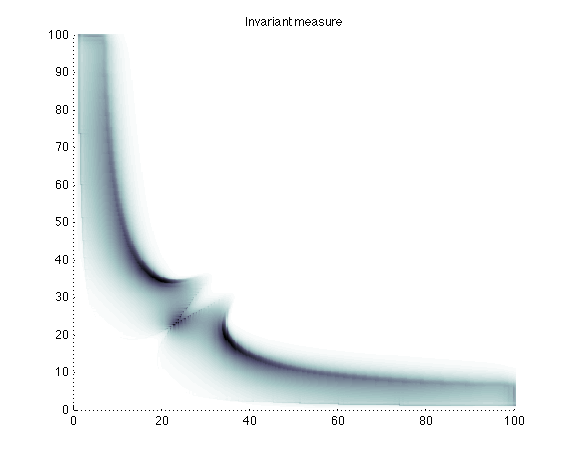}} %
\caption{Comparison of invariant measures:effect of having more impatient Institutional Investors and more impatient HFT (tests 4 and 6).}
\label{fig:HFTII_2}
\end{figure} 
 
\paragraph{Test 5.} Our aim is to model a market opened to HFT, and to observe the effects of the arrival of HFTs. To do so, we consider the following case:
\begin{itemize}
\item HFT order sizes are four times smaller than Institutional Investors orders.
\item All HFT arbitrate between being Liquidity Providers or Liquidity Consumers while half of Institutional Investors are Liquidity Consumers anyway.
\end{itemize}
The selected parameters corresponding to such a market are:
\begin{itemize}
\item General parameters : $\delta = 2, P=100$.
\item Institutional Investors: $\;\Lambda_1= \lambda_1+\lambda^-_1 = \frac 1  2 + \frac 1 2,\; q_1 = 1, \; c_1 = 0.25\%$.
\item HFT: $\Lambda_2= \lambda_2+\lambda^-_2 = 4 + 0, \; q_2 = 0.25, c_2 = 1\%$
\end{itemize}
  
\noindent
Remind that we have chosen to set $c_1 q_1 = c_2 q_2$.\\
Another modeling assumption is that orders are equally split between both types, that is $\Lambda_1 = \Lambda_2$.

\bigskip
\noindent
As in section \ref{sec:FOA}, we can distinguish several cases depending upon the action type of the traders. This lead to 9 distinct regions in the $Q_a \times Q_b$ space,  which are depicted in Figure \ref{fig:9reg}. Note that the upper-index denotes the action ($+$ stands for Liquidity Provider, $-$ for Liquidity Consumer) and the lower-index denotes the type of trader. 
  
We can compare the two following situations: the market stationary equilibrium with a single specie of traders (Institutional Investors), and the market after the arrival of HFTs. \\
Figure \ref{fig:HFTII} shows the corresponding stationary measure of states (size of ask and bid queues). \\
We notice that in the case with Institutional Investors only, there are stable liquidity imbalances with two symmetric configurations (one favorable to buyers, and the other one to sellers). On the other hand, in the case with both Institutional Investors  \& HFT, we observe a liquidity stabilization and a concentration on a single balance equilibrium.

\bigskip
\noindent
Recall that in the present case with $c_1q_1 = c_2q_2$, the value functions of Institutional Investors and HFT coincide. Consequently the existence of nine regions and the stabilizing effect described above are only explained by market impact heterogeneity. Which is a noteworthy numerical result.

\bigskip
\noindent
In Table 4 we display the numerical values of the average transaction prices (only for sellers since the prices for buyers are symmetrical, the \emph{fair price} being 100).

\noindent
We remark that HFT trade at better prices than Institutional Investors and that Institutional Investors average selling price decreases in the market with HFT. Consequently, in this case HFT traders capture the difference, and even more.

\begin{table}[h!]\label{tabep}
\begin{center}
\begin{tabular}{|c|c||c|c|c|} \hline 
 & Test 4 & \multicolumn{3}{|c|}{Test 5} \\\cline{2-5}
& \mbox{II only}  & \mbox{II in the mix} & \mbox{HFT in the mix} & \mbox{mix}  \\\hline
 \mbox{Liquidity Consumers} & 99.849 & 99.842 & 99.938 & 99.89 \\\hline
 \mbox{Liquidity Providers} &100.238 &100.103 &100.189 &100.146 \\\hline
 \mbox{Average} & 99.876 & 99.852 & 99.981 & 99.916 \\\hline
\mbox{Spread $\psi$}& 0.248 & 0.296 & 0.038 & 0.167 \\\hline
\mbox{Spread $\psi$ (bps)}& 25&	30&	4&	17 \\\hline
\end{tabular}
\end{center}
\caption{Expected transaction prices and spread in the model}
\end{table}

\noindent
The last row shows the impact on the spread. In the framework of our model we define the expected bid-ask spread as the difference $\psi=\E(\PB)-\E(\PS)$.
We conclude that the spread increase for Institutional Investors is 20\%, while the global spread decrease of the market is worth 33\%. Consequently, the spread reduction clearly profits to HFTs.

Our simulation results are compatible with spread shrinking scenarios in which while reducing the spread, HFT provoke more impatience among other (slower) types of investors. The result of such a mixing is that fast agents, reducing the spread, are more passive than other class of investors, or aggressive when the spread is at their advantage.

\paragraph{Comparison with empirical data.}
\cite{citeulike:12303370} uses real data to fit observed order flows with coef of a Fokker-Planck like PDE, compatible with our MFG approach. The Figure 9 of this empirical paper exhibits configurations that are close to our theoretical stationary results.

\paragraph{Test 6.} For the sake of completeness we end this section with another example, for other proportions of SOR and non-SOR traders. Here we take the same parameters as in Test 1 except that $\Lambda_1= \lambda_1+\lambda^-_1 = 90\% \times 4 + 10\% \times 4$ and $\Lambda_2= \lambda_2+\lambda^-_2 = 60\% + 40\% $.
Thus we look at a situation where there are 10 points more SOR in the Institutional Investor population and 10 points less SOR amongst HFTs.
In Figure \ref{fig:HFTII_2} we compare the situation between a market with Institutional Investors only and with a mix of both Institutional Investors and HFTs. \\
Table 5 shows the corresponding quantities.

\begin{table}[h!]\label{tabep}
\begin{center}
\begin{tabular}{|c|c||c|c|c|} \hline 
 & Test 4 & \multicolumn{3}{|c|}{Test 6} \\\cline{2-5}
& \mbox{II only}  & \mbox{II in the mix} & \mbox{HFT in the mix} & \mbox{mix}  \\\hline
 \mbox{Liquidity Consumers} & 99.898&	99.854&	99.994&	99.924 \\\hline
 \mbox{Liquidity Providers} & 100.168	&100.094&	100.157&	100.125 \\\hline
 \mbox{Average} & 99.91	&99.864&	99.974&	99.919 \\\hline
\mbox{Spread $\psi$}& 0.180 & 0.272 & 0.052 & 0.162 \\\hline
\mbox{Spread $\psi$ (bps)}&18&	27	&5&	16 \\\hline
\end{tabular}
\end{center}\caption{Expected transaction prices and spread in the model}
\end{table}

\section{Conclusion}

This paper demonstrates how MFG (Mean Filed Games, \cite{citeulike:3614137}) can be used to model orderbook dynamics. At the junction of structural approaches (see \cite{citeulike:6253089} and \cite{fou05}) and flow driven ones (see \cite{citeulike:12303370}), the mean field game render the strategic behaviour of traders, leading to partial derivative equations that can be numerically solved, and partly reduced to simpler dynamics (Section \ref{sec:MFG} introduces mean field games).

The application presented here used a stylized orderbook model in which:
\begin{itemize}
\item each side of the orderbook (buy or sell) is captured by one variable: its size (i.e. the number of orders waiting in each queue);
\item the ``fair'' or ``latent'' price is stable, since we focus on microstructure effects on the traded price;
\item the market impact of a trade is close to linear;
\item the orderbook has an infinite resiliency, in the sense that its shape does not change through time: when liquidity is consumed in one queue, its shape readjusts immediately to a linear one (even if its size changes).
\end{itemize}
The trading strategies of the investors are described by:
\begin{itemize}
\item their arrival rate $\lambda$, following an homogenous Poisson process;
\item the average size of their orders $q$;
\item their waiting cost $c$: the larger they are, the more impatient the investor;
\item a fraction of the orders of investors have an infinite impatience: we call them ``non SOR (Smart Order Router) users'' since they are not patient enough to follow a liquidity-driven microscopic trading strategy.
They can be considered as having an infinite waiting cost.
\end{itemize}

Section \ref{sec:pfp} studies the dynamics of such a model and Section \ref{sec:IIonlysims} shows that such a stylized modelling give birth to realistic dynamics: with one class only of investors, stable states of liquidity imbalance can appear. 
This can be read as a justification for the introduction of the role of market makers.

In a third step of our reasoning, we introduce HFT (High Frequency Traders) with the hope they will assume this market making role. Consistently with \cite{citeulike:8423311} and \cite{citeulike:11858957}, they are modelled as: fast, using smaller orders than institutional investors, and taking decisions according to the immediate state of the orderbook (in our vocabulary, they are ``Smart Order Router''  users). It is important to underline that they have not a different impatience (i.e. waiting cost per share) than other investors.

Section \ref{sec:twogroupseq} extends the approach developed in Section \ref{sec:pfp}  to a model with two types of investors (to be applied to institutional ones and HFT). We then study numerically the properties of markets with institutional investors and HFTs, looking for an answer to regulators and policy makers questions about the effect of mixing two so different classes of market participants.
First note that our results in terms of invariant measure distribution are consistent with data explorations conducted in \cite{citeulike:12303370}.

Qualitatively, our conclusions are that \emph{the introduction  of HFT improves the usual measures of the efficiency of the price formation process}: the stable states of offer and demand are more balanced and the effective bid-ask spread is smaller  than without HFTs.
But \emph{the observed improvement is at the exclusive advantage of the HFTs}: the effective bid-ask spread paid by institutional investors is largest than before the introduction of HFT.
Of course these conclusions are conditioned to the accuracy of our stylized model; nevertheless they can explain the disjunction between the claims of institutional investors (that, for them, the price formation process is more difficult to deal with in presence of HFTs) and the objective improvement of measurements of the state of liquidity since HFT activity increased.

Hence this paper is not only a contribution to the modelling of orderbook dynamics, showing how a MFG-approach can conciliate structural and flow-driven approaches. It provides a qualitative analysis of the role of High Frequency Trading in electronic markets. It also underlines the lack of liquidity measurements adapted to the current market microstructure.


\appendix

\section{Proofs}

\subsection{Proof of Proposition \ref{prop:M0}}
\label{sec:proofM0}


Looking at equations $(A_{R^{++}})$ and $(A_{R^{-+}})$ we notice that at the boundary there is a jump causing a change of sign of the coefficient multiplying the derivatives (under the basic assumption $\lambda \geq \lambda^-$). Therefore, at this point we must have

\begin{equation}\label{eq:firstc}
\begin{split}
\frac {\lambda^-} {x_0} (\frac {\delta}{x_0-q}-\tilde u) = \frac{c}{q}.
\end{split}
\end{equation}
On the other hand, as the seller's routing decision $R^\oplus_s$ jumps from 1 to 0, we must have
\begin{equation}\label{eq:secondc}
\begin{split}
\tilde u = \frac{- \delta}{y_0-q}.
\end{split}
\end{equation}
Combining (\ref{eq:firstc}) and (\ref{eq:secondc}) we get the equality:
\begin{equation*}
\eta x_0 = \frac{1}{x_0-q}+\frac{1}{y_0-q},
\end{equation*}
where (same definition as in Proposition \ref{prop:M0}):
$$\eta := \frac{c}{ \delta q\lambda^-}.$$
It follows that the diagonal point of the boundary $M_0$ is the point (equation \ref{x0diag} of Proposition \ref{prop:M0})
$$(x_0^*,x_0^*) = (q+\sqrt{q^2+8/\eta})/2$$
and that the boundary is defined by the parametric equation of Proposition \ref{prop:M0}:
$$(x_0,y_0) = \Big(x_0,l(x_0):= q+ \Big(\eta x_0-\frac 1 {x_0-q}\Big)^{-1}\Big), \; \forall x_0 \geq x_0^*.$$
\hfill\cqfd~\\

\subsection{Proof of Proposition \ref{prop:M1}}
\label{sec:proofM1}


Unfortunately, looking at equations $(B_{R^{-+}})$ and $(B_{R^{--}})$ we conclude that we cannot adopt the same reasoning since the sign of the coefficients multiplying the derivative terms does not change.\\
We use another strategy. We solve $\tilde v$ analytically all along the characteristic line $y_1=x_1-k$, and then intersect the solution $\tilde v$ with $\frac{\delta}{x_1-q}$.

\bigskip
\noindent
Along the characteristic $x=y+k$, we introduce the function \[ f(y) = \tilde v(y+k,y).\]
Looking at equation (\ref{eq:pden}), we get the generic form of the ordinary differential equation (ODE for short) satisfied by $f$:
\begin{equation}\label{eq:ode}
f'+\frac a y f+\Big( \frac b {y(y-q)}+d\Big)=0,
\end{equation}
where
\begin{equation*}
a =1+\lambda/\lambda^-, \; b = \delta (1+\lambda/\lambda^-), \;  d= -\delta \eta, \; \mbox{on } R^{-+}.
\end{equation*}
We use the variation of constant method to solve equation (\ref{eq:ode}).\\
The homogeneous solution is $f(y)=y^{-a}$ times a constant.
Now let the constant varies as a function $g(x)$. We have $f' = g'y^{-a}-agy^{-a-1}$. Substituting in (\ref{eq:ode}) we obtain:
\[g' (y)= -b \frac{y^{a-1}}{y-q}-d y^a.\] This function is easy to integrate numerically. However in order to stay working with analytical formulas, we make the approximation $y-q \approx y$ for small $q$ (recall that all this analytical part focus on the small $q$ first order approximation). Now we are in the position to integrate the derivative $g'$.\\
We get \[f(y) = g(y) y^{-a} = (\kappa \; y^{-a}-\frac b {a-1} y^{-1} - \frac d {a+1} y) .\]

\bigskip
\noindent
Now we have to compute the constant $\kappa$. Recall that we are working on the line $(y+k,y)$ and on the region $R^{-+}$ so that we are solving the ODE with an initial condition on $M_0$, which is known to be $(x_0,l(x_0))$.\\
Consequently we have to look at $f$ as a family $(f_k)$ of functions indexed by $k \in \mathbb{R}^+$. On the characteristic line starting at $x_0-l(x_0)$, the function is given by
\begin{equation}\label{eq:f}
f_{x_0-l(x_0)}(y) = (\mathcal{C} (x_0) \; y^{-a}-\frac b {a-1} y^{-1} - \frac d {a+1} y), \; \forall y \geq l(x_0) .
\end{equation}
The core argument to compute the constant parameter $\mathcal{C} (x_0)$ for the solution on the characteristic $(y+k,y),$ with $k=x_0-l(x_0)$, is to remark that:

\[f_k(y) = \tilde v(y+k,y) = -\tilde u(y,y+k) =-f_k(y+k).\]
Then, the initial condition equality  \[f_{x_0-l(x_0)}(l(x_0))=-f_{x_0-l(x_0)}(x_0),\]
automatically gives the expression of $\mathcal{C}$:
\begin{equation}\label{eq:C}
\mathcal{C} (x_0) = \delta  \frac{(1+\lambda^-/\lambda) [x_0^{-1}+l(x_0)^{-1}]-\frac{\eta}{1+\lambda/\lambda^-} [x_0+l(x_0)]}
				      		 {x_0^{-(1+\lambda/\lambda^-)}+l(x_0)^{-(1+\lambda/\lambda^-)} } ,
\end{equation}

\noindent
where the last equality holds since the equation of $\tilde u$ on $R^{+-}$ matches the equation of $\tilde v$ on $R^{+-}$.\\
Consequently, the analytical solution is given by (\ref{eq:f})-(\ref{eq:C}).

\bigskip
\noindent
Finally we are in the position to compute the parametric curve of the boundary between the two regions $R^{-+}$ and $R^{--}$.\\
To do so we look for the point $(x_1,y_1) = (y_1+k,y_1)$ such that $\tilde v(x_1,y_1) = \frac \delta {x_1-q} $. \\
More precisely, $M_1$ is defined by: $(y_1+x_0-l(x_0),y_1), \; \forall x_0 \geq x_0^*,$
where (equation \ref{eq:boundM1} of Proposition \ref{prop:M1})
\begin{equation*}
y_1 \mbox{ verifies } f_{x_0-l(x_0)}(y_1) = \frac \delta {y_1+x_0-l(x_0)-q}.
\end{equation*}
\hfill\cqfd~\\

\subsection{Local equations of the four regions (second order equations)}
\label{sec:regions2nd}

Define $\Lambda = \lambda+\lambda^-$. Let us now give the local equations on the same four regions.

\begin{equation}\label{eq:secondordereqs}
\begin{split}
(A_{R^{++}}) \;\;\;0 &= \Big[\frac {\lambda^-} x (p^b(x) -  u) - c \Big]  + [\lambda  -\lambda^-]  (\partial_x   u + \partial_y   u) + q\Big(\frac{\lambda^-}{x } \partial_x   u + \frac{\Lambda}{2} \Delta   u\Big),\\
(B_{R^{++}}) \;\;\;0 &= \Big[\frac {\lambda^-} y( p^s(y)  -  v) + c \Big] + [\lambda  -\lambda^-]  (\partial_x  v+  \partial_y  v)+ q\Big(\frac{\lambda^-}{y } \partial_y   v + \frac{\Lambda}{2} \Delta   v\Big), \\
(A_{R^{-+}}) \;\;\;0 &= \Big[\frac {\lambda^-} x (p^b(x) -  u) - c \Big]  + [  -\lambda^-]  (\partial_x   u + \partial_y   u) + q\Big(\frac{\lambda^-}{x } \partial_x   u + \frac{\lambda^-} 2 \Delta   u+ \lambda \partial_{yy}   u \Big) ,\\
(B_{R^{-+}}) \;\;\;0 &= \Big[\frac {\Lambda} y( p^s(y)  -  v) + c \Big] + [  -\lambda^-]  (\partial_x  v+  \partial_y  v)
+ q\Big(\frac{\Lambda}{y } \partial_y   u + \frac{\lambda^-} 2 \Delta   u+ \lambda \partial_{yy}   v \Big) , \\
(A_{R^{+-}}) \;\;\;0 &= \Big[\frac {\Lambda} x (p^b(x) -  u) - c \Big]  + [  -\lambda^-]  (\partial_x   u + \partial_y   u) + q\Big(\frac{\Lambda}{x } \partial_x   u + \frac{\lambda^-} 2 \Delta   u+ \lambda \partial_{xx}   u \Big) ,\\
(B_{R^{+-}}) \;\;\;0 &= \Big[\frac {\lambda^-} y( p^s(y)  -  v) + c \Big] + [  -\lambda^-]  (\partial_x  v+  \partial_y  v)
+ q\Big(\frac{\lambda^-}{y } \partial_y   u + \frac{\lambda^-} 2 \Delta   u+ \lambda \partial_{xx}   v \Big) , \\
(A_{R^{--}}) \;\;\;0 &= \Big[\frac {\Lambda} x (p^b(x) -  u) - c \Big]  + [-\Lambda]  (\partial_x   u + \partial_y   u) + q\Big(\frac{\Lambda}{x } \partial_x   u + \frac{\Lambda}{2} \Delta   u\Big),\\
(B_{R^{--}}) \;\;\;0 &= \Big[\frac {\Lambda} y( p^s(y)  -  v) + c \Big] + [-\Lambda]  (\partial_x  v+  \partial_y  v)+q\Big(\frac{\Lambda}{y } \partial_y   v + \frac{\Lambda}{2} \Delta   v\Big).
\end{split}
\end{equation}

Where $\Delta$ stands for the Laplacian operator:
$$\Delta f=\partial_{xx} f + \partial_{yy} f.$$

Remark that, compared to equations (\ref{eq:firstordereqs}), both a diffusion term {\bf and a drift term} appear.

\bibliographystyle{apalike}
\bibliography{lehalle}

\end{document}